\documentclass[aps,pra,twocolumn,amsmath,amssymb,nofootinbib,showpacs,superscriptaddress]{revtex4-1}
\usepackage[english]{babel}
\usepackage{latexsym}
\usepackage{graphics}
\usepackage{graphicx}
\usepackage{epsfig}
\usepackage{color}
\usepackage{bm}
\usepackage{amsmath}
\usepackage{amssymb}
\usepackage{amsthm}
\usepackage{dcolumn}
\usepackage{bm}
\usepackage{float}
\usepackage{hyperref}
\usepackage{color}
\usepackage{epstopdf}
\usepackage{cleveref}
\usepackage[svgnames]{xcolor}
\usepackage[normalem]{ulem}
\usepackage{enumitem}

\hypersetup{hidelinks,colorlinks=true,allcolors=DarkBlue}

\theoremstyle{remark}

\newtheorem{definition}{Definition}
\newtheorem{theorem}{Theorem}
\newtheorem{lemma}{Lemma}

\begin{document}

\preprint{APS/123-QED}

\title{Practical quantum multiparty signatures using quantum-key-distribution networks}

\author{E.O. Kiktenko}
\affiliation{Steklov Mathematical Institute, Russian Academy of Sciences, Moscow 119991, Russia}
\affiliation{Russian Quantum Center, Skolkovo, Moscow 143025, Russia}
\affiliation{QApp, Skolkovo, Moscow 143025, Russia}
\affiliation{Moscow Institute of Physics and Technology, Moscow Region 141700, Russia}

\author{A.S. Zelenetsky}
\affiliation{Russian Quantum Center, Skolkovo, Moscow 143025, Russia}
\affiliation{QApp, Skolkovo, Moscow 143025, Russia}
\affiliation{Bauman Moscow State Technical University, Moscow 105005, Russia}
\author{A.K. Fedorov}
\affiliation{Russian Quantum Center, Skolkovo, Moscow 143025, Russia}
\affiliation{QApp, Skolkovo, Moscow 143025, Russia}
\affiliation{Moscow Institute of Physics and Technology, Moscow Region 141700, Russia}

\date{\today}
\begin{abstract}
Digital signatures are widely used for providing security of communications.
At the same time, the security of currently deployed digital signature protocols is based on unproven computational assumptions.
An efficient way to ensure an unconditional (information-theoretic) security of communication is to use quantum key distribution (QKD), whose security is based on laws of quantum mechanics. 
In this work, we develop an unconditionally secure signature scheme that guarantees authenticity and transferability of arbitrary length messages in a QKD network. 
In the proposed setup, the QKD network consists of two subnetworks: (i) an internal network that includes the signer and with limitation on the number of malicious nodes and (ii) an external network that has no assumptions on the number of malicious nodes.
A consequence of the absence of the trust assumption in the external subnetwork is the necessity of  assistance from internal subnetwork recipients for the verification of message-signature pairs by external subnetwork recipients.
We provide a comprehensive security analysis of the developed scheme, perform an optimization of the scheme parameters with respect to the secret key consumption, 
and demonstrate that the developed scheme is compatible with the capabilities of currently available QKD devices.
\end{abstract}

\maketitle

\section{Introduction}

An essential task for modern society is to guarantee the identity of a sender and the authenticity of a message within electronic communications.
This problem can be solved with the use of digital signatures~\cite{DH1976}.
Importantly, digital signatures also guarantee that messages are transferable, so a forwarded message could also be accepted as valid.
Currently deployed digital signatures are mostly based on unproven computational assumptions, such as the computational complexity of factoring large integers or computing discrete logarithms.
This task is believed to be computationally hard for classical computers, but it appeared to be solved in polynomial time using a large-scale quantum computer~\cite{Shor1997}.
This has stimulated active research on the possibility of realizing digital signature, which are resistant to attacks with quantum computers.

One particular option is to use quantum signatures that provide an unconditional (information-theoretic) level of security.
In the seminal theoretical proposal~\cite{GottesmanChuang2001} a quantum version of Lamport's one-time signature scheme~\cite{Lamport1979} based on a one-way quantum function was considered.
Alternative scheme that is based on a quantum one-way function and the involvement of a trusted party was proposed in Ref.~\cite{Lu2005}.
Theoretical proposals were followed by the first experimental demonstration reported in Ref.~\cite{Clarke2012}.
Unfortunately, all these schemes require efficient quantum memory, which is still at an immature stage in its technology. 
The important step towards developing quantum signatures is removing the demanding requirement of quantum memory considered both theoretically~\cite{Andersson2006,Dunjko2014,Wallden2014,Yin2016,Amiri2015,Amiri2016} and experimentally~\cite{Dunjko20142,Donaldson2016,Roberts2017,Zhang2018,Ding2020,Zhang2021}.

An important class of quantum signature schemes (see  Refs.~\cite{Wallden2014,Roberts2017}) is the one based entirely on the technology of quantum key distribution (QKD)~\cite{Gisin2002,Scarani2009,Lo2015,Lo2016}, which is currently available at the commercial level.
In a sense, this approach follows the development of `classical' (traditional) unconditionally secure signature scheme~\cite{Chaum1991, Pfitzmann1996, Hanaoka2000, Hanaoka2004, Shikata2002, Swanson2011,Amiri2015} that provide authenticity and transferability of (classical) messages based on some resource such as authenticated broadcast channel and secret authenticated classical channels.
Since the QKD provides legitimate parties with unconditionally secure symmetric keys, any unconditionally secure signature (USS) scheme that requires secret authenticated classical channels appears to be suitable for implementation in contemporary QKD networks.

However, several obstacles prevent the practical deployment of QKD-assisted USS schemes~\cite{Wallden2014,Roberts2017} as well other quantum signature schemes~\cite{Dunjko2014,Amiri2016,Dunjko20142,Donaldson2016}.
The first one is that all these schemes are analyzed for a network consisting of three parties only.
The second one is that messages of the length of only a single bit are considered.
These two issues are covered in Refs.~\cite{Arrazola2016,Amiri2018}: Specifically, in Ref.~\cite{Arrazola2016} a multiparty QKD-assisted USS scheme was introduced, while Ref.~\cite{Amiri2018} proposed employing almost strongly 2-universal families to sign messages of practically arbitrary length.
Nevertheless, as we show in our work, these schemes suffer from security loopholes that appear in the realization of these schemes in practice.
Moreover, the security bounds derived in security proofs in ~\cite{Amiri2018} are based on the simplified case of two-bit authentication tags and though they demonstrate good asymptotic behavior of secret key consumption, they appear to be impractical for deployment in realistic QKD networks.

In the present work we improve the practicality of USS schemes. 
By revising the results of Ref.~\cite{Amiri2018}, we develop an unconditionally secure signature  scheme that guarantees authenticity and transferability of practically arbitrary length messages in a QKD network of more than four nodes.
In contrast to previous designs of USS schemes, we consider the global network consisting of two subnetworks (see Fig.~\ref{fig:network}): the internal network that includes the signer and with upper bound $\omega$ on the number of malicious nodes and the external network that has no assumptions on the number of dishonest users (in these terms, previously proposed schemes consider the internal subnetwork only).
We introduce a concept of the delegated verification that allows external recipients to verify and forward message-signature pairs by means of assistance from internal subnetwork recipients.
We provide a full security analysis of the developed scheme, perform numerical optimization of the scheme parameters with respect to the secret key consumption, and demonstrate that the developed scheme is compatible with the capabilities of contemporary QKD devices.

\begin{figure}
    \centering
    \includegraphics[width=\linewidth]{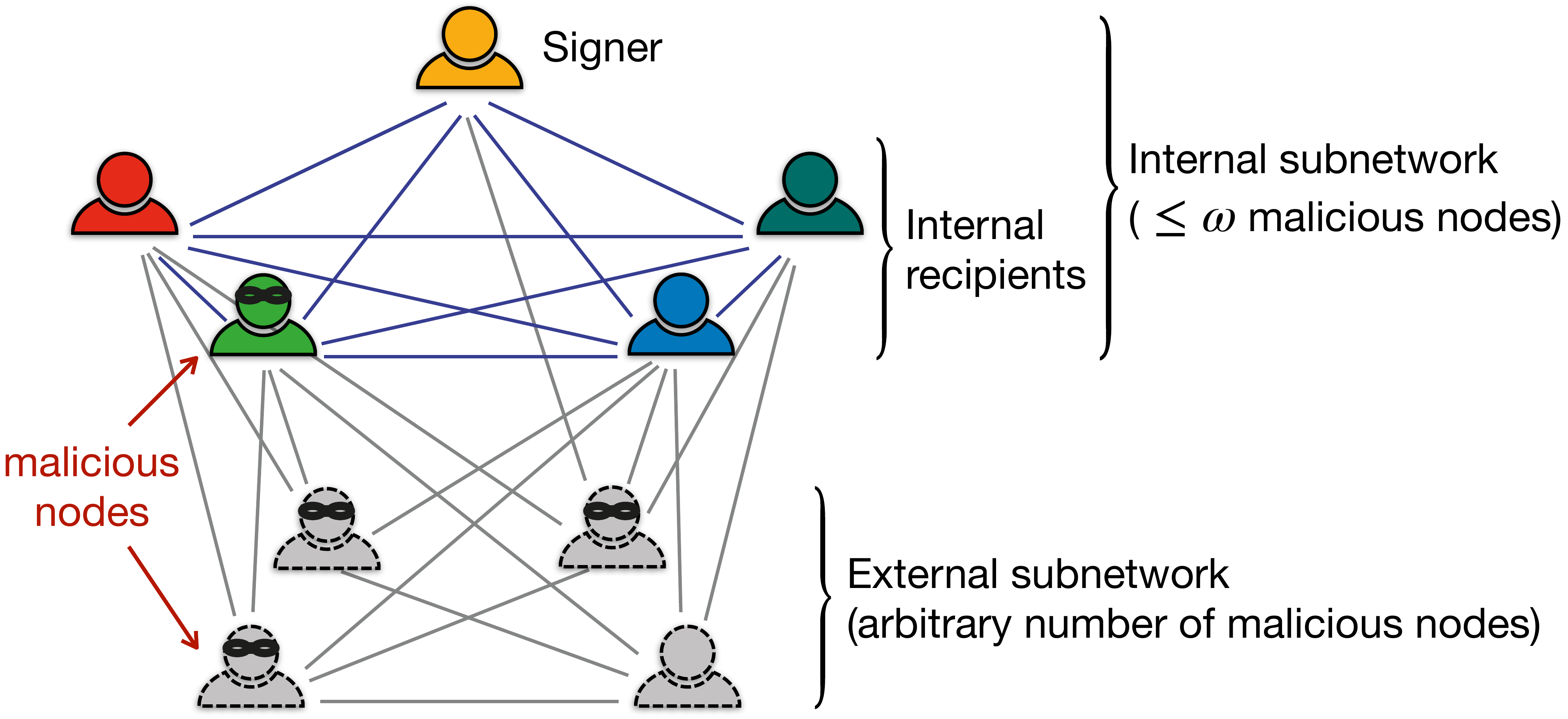}
    \caption{Organization of the QKD network used in the developed USS scheme. 
    The signer and the set of $N$ internal recipients are connected within an all-to-all topology. 
    Each of the $M$ external recipients is necessarily connected with at least $2\omega+1$ internal recipients.
    The requirement on the secret key generation rate for the QKD links connecting nodes of the internal subnetwork is higher compared to all other QKD links.
    }
    \label{fig:network}
\end{figure}

Our work is organized as follows.
In Sec.~\ref{sec:scheme} we describe a general scheme of the proposed QKD-assisted USS scheme.
In Sec.~\ref{sec:security} we introduce the main security definitions and formulate security statements.
In Sec.~\ref{sec:performance} we analyze the performance of the proposed quantum digital signature scheme.
We summarize our results and conclude in Sec.~\ref{sec:conclusion}.

\section{QKD-assisted USS scheme: \\ The workflow}
\label{sec:scheme}

\subsection{Organization of the network}

We consider an operation of the USS scheme in a QKD network consisting of $N+M+1$ nodes, where $N\geq 4$ and $M\geq 0$ (see Fig.~\ref{fig:network}).
Here, by the QKD network, we understand a set nodes (devices) which are able to communicate with classical messages and also are connected with pair-wise QKD setups producing unconditionally secure symmetric keys.
We refer to the whole network of $N+M+1$ nodes as the global QKD network.
Within this global network we distinguish two subnetworks: (i) a set of $N+1$ nodes labeled by $\{{\cal P}_i\}_{i=0}^{N}$ and called an internal subnetwork and (ii) the remaining part of the global network consisting of $M$ nodes labeled by $\{{\cal E}_i\}_{i=1}^{M}$ and called an external subnetwork.
In our USS scheme a message is signed by a distinguished node ${\cal P}_0$, named the signer.
Other nodes $\{{\cal P}_i\}_{i=1}^N$ and $\{{\cal E}\}_{i=1}^M$ are called internal and external recipients,  respectively.

The difference between internal and external subnetworks, in addition to the fact that the signer belongs to the internal one, is in (i) the trust assumptions, (ii) the connectivity, (iii) the symmetric secret key consumption, and (iv) the general verification principle.
Below we discuss all these points in detail.

\emph{Trust assumptions.} We assume that the number of dishonest (malicious) internal nodes does not exceed some positive integer $\omega$, which is one of the basic parameters of our scheme.
In the next section we introduce a strict bound on $\omega$,  but for now we note that $\omega$ is definitely less than $N/3$.
In contrast to the internal subnetwork, the number of dishonest nodes in the external subnetwork is unbounded.
In this way, the global network can be considered as a moderately trusted internal subnetwork surrounded by an untrusted external one.

\emph{Connectivity.}
All the nodes in the internal network (the signer and internal recipients) have to be connected with each other via pairwise QKD setups in an all-to-all fashion.
All that is required from external recipients is to be connected by pairwise QKD setups with at least $2\omega+1$ internal recipients.
Other connections between nodes are of course possible, but are not necessary for the operation of the scheme.

\emph{Symmetric key consumption.}
In our scheme, unconditionally secure symmetric keys obtained from QKD links are employed for two basic purposes: (i) providing unconditionally secure symmetric encryption with one-time pads (OTPs) and (ii) providing unconditionally secure authentication with almost strongly 2-universal (AS2U) family of functions~\cite{Wegman1981}.
We note that although the whole mechanism of the considered USS scheme is based on AS2U families as well, we separate the goals of providing authenticity of pairwise classical channels used in communication between nodes and providing transferability of signatures.
Next we assume that all messages between nodes in the global network are transmitted via pairwise perfect authenticated channels.
We discuss a practical justification of this assumption in more detail in Sec.~\ref{sec:performance}.
The crucial point is that in our scheme the OTP encryption is used only in a communication between nodes of the internal subnetwork.
In Sec.~\ref{sec:performance} we show that the corresponding key consumption turns out to be much higher than the one for providing unconditionally secure authentication. 
That is why secret key rate requirements for QKD links connecting internal with external and external with external recipients appear to be weak compared to the links between internal subnetwork nodes.

\emph{Verification principle.}
The verification of a message-signature pair, produced by ${\cal P}_0$, and its forwarding within the global network can be performed by both internal and external recipients; however, the corresponding verification processes are different.
Internal recipients are able to verify received message-signature pairs directly without additional communication with other nodes. 
External recipients have to communicate to $2\omega+1$ internal recipients for verification purposes.
We call this verification process run by external nodes delegated verification.

\subsection{Basic idea of the scheme}

Before going into the technical description of the scheme, here we sketch the general principle of the USS scheme operation.
It is based on employing special `asymmetric' keys distributed among the internal subnetwork.
In contrast to traditional digital signatures schemes, where the signer possesses a secret (private) key and recipients possess a common public key, in the USS scheme, all internal recipients possess different keys that have to be kept secret as well.
Next we refer to the key owned by the signer and used for generating signatures as the signing key, while the keys possessed by each internal recipient, used for validating signatures, are called verification keys.
The external recipients do not have any special verification keys and they seek assistance from internal recipients to verify a signature, that is, delegate the verification.
The additional communication in the delegated verification can be considered as a consequence of the lack of any trust assumption with respect to the external recipients and milder conditions on their connectivity.

We separate the workflow of the USS into two basic stages: 
(i) the preliminary distribution stage, where signing and verification keys are distributed throughout the internal network, and 
(ii) the main messaging stage, at which the signer generates a signature for some message and transmits the message-signature pair to (some) internal or external recipients, who then are able to forward this message-signature pair to each other.
Stage (ii) can also include a special majority vote dispute resolution process that is a consequence of the finite transferability of the considered USS scheme.

In our work we consider a one-time scenario, where each distributed set of keys provides security for a single message-signature pair only.
Surely, the developed scheme can be applied for multiple messages by means of parallel communication.

As already mentioned, our scheme is based on employing AS2U families that are commonly used for providing unconditionally secure authentication given that the signer and recipient share symmetric keys, which we refer to as authentication keys.
These keys are used to choose a function from an AS2U family and compute the corresponding authentication tag that is an output of the chosen function for a given message to be sent.
The general idea behind the considered USS scheme is that both the signing and verification key consist of a number of authentication keys and
the message's signature consists of a number of authentication tags computed with different authentication keys.
Each of the internal recipients knows only some of the authentication keys and therefore is able to verify the signature but is not able to forge a signature for some alternative message.
At the same time, the signer knows all the authentication keys but does not know a particular subset of keys possessed by a particular recipient.
This condition is necessary to avoid non-transferability and repudiation of a generated message-signature pair.  
The required uncertainty in keys is achieved by first transferring different authentication keys (different parts of the signing key) from the signer to each of the internal recipients and then by random shuffling of the obtained keys between internal recipients in secret from the signer.
The idea behind this delegated verification is that in order to verify a message-signature pair, an external recipient communicates to $2\omega+1$ internal recipients, the majority of whom are definitely honest. 
So the result of the verification can be determined by the result of verification obtained by the majority of requested internal recipients.

It what follows we consider the workflow of the scheme in detail. 
We note that, in general, the same notation as in Ref.~\cite{Amiri2018} is used.
The main differences in the workflow of our scheme compared to the proposal in Ref.~\cite{Amiri2018} are summarized in Appendix~\ref{sec:app:differences}.
We also provide a list of all the basic parameters of the scheme in Tab.~\ref{tab:vars}.

\begin{figure*}[]
   \centering
    \includegraphics[width=0.8\linewidth]{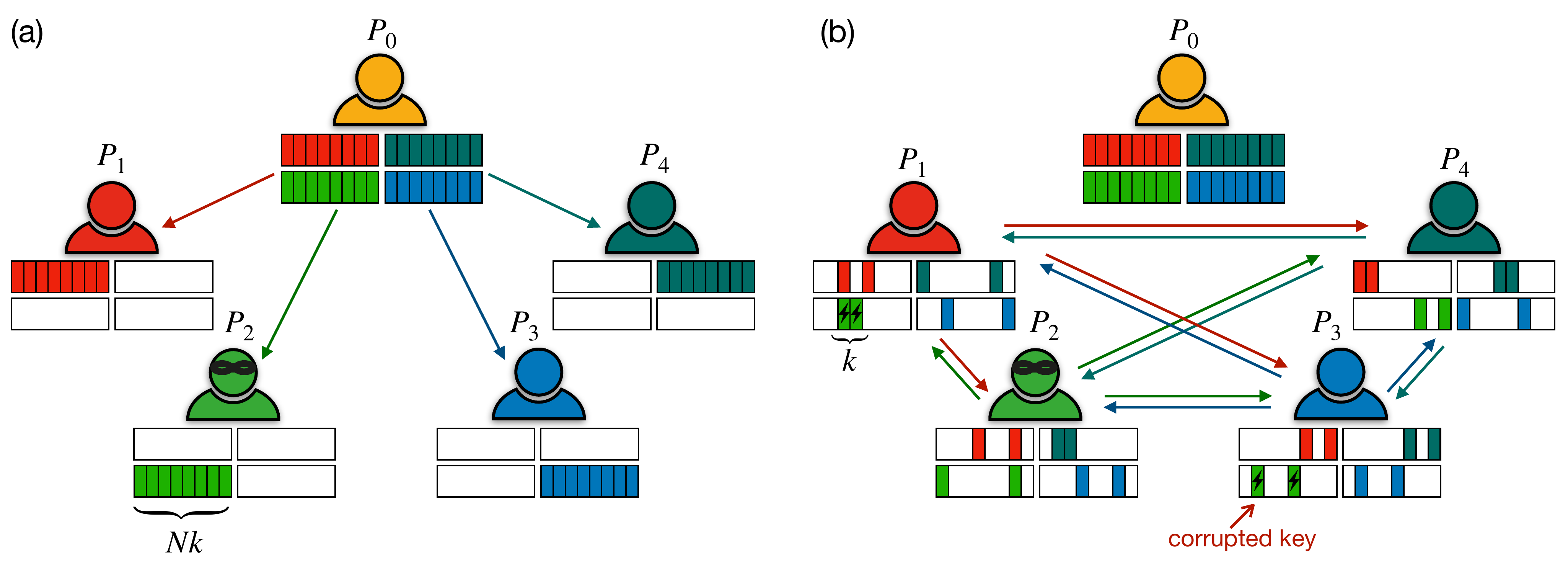}
    \caption{
    Scheme of the distribution stage run within the internal subnetwork.
    (a) At step 1 the signer $\mathcal{P}_0$ generates $N^2k$ authentication keys, which define elements from the AS2U family, and transmits $Nk$ keys to each of the recipients.
    (b) Then at step 2 each recipient $\mathcal{P}_i$ randomly splits the obtained keys into $N$ disjoints chunks of $k$ keys,  sends $N-1$ chunks accompanied by the corresponding key indices (in respect to the original signing key) to other recipients (one chunk for each recipient), and leaves a single chunk for oneself.
    All the communication is secured by the OTP encryption supported by pairwise QKD links. 
    Here ${\cal P}_2$ is dishonest and sends corrupted keys to ${\cal P}_1$ and ${\cal P}_3$ at the second step.}
    \label{fig:distribution}
\end{figure*}

\subsection{Preliminary distribution stage}

We begin with a formal definition of the AS2U family that is the basis of the whole QKD-assisted USS scheme.
\begin{definition}[AS2U family].
	Let $\mathcal{A}$, $\mathcal{B}$, and $\mathcal{K}$ be finite sets. 
	A family of functions $\mathcal{F} = \{f_\kappa : \mathcal{A} \rightarrow \mathcal{B}\}_{\kappa\in\mathcal{K}}$ is called $\varepsilon$-almost strongly 2-universal ($\varepsilon$-AS2U) if the two following requirements are satisfied.
	\begin{enumerate}[label=(\roman*)]
		\item For any $m \in \mathcal{A}$ and $t \in \mathcal{B}$ one has
		$\Pr[f_\kappa(m) = t] = |\mathcal{B}|^{-1}$ for $\kappa$ picked uniformly at random from $\mathcal{K}$.
		\item For any distinct $m_1, m_2 \in \mathcal{A}$ and any $t_1, t_2 \in \mathcal{B}$ one has $\Pr[f_\kappa(m_1) = t_1 | f_\kappa(m_2) = t_2] \leq \varepsilon$ for $\kappa$ picked uniformly at random from $\mathcal{K}$.
    \end{enumerate}
    If $\epsilon=|\mathcal{B}|^{-1}$, then $\mathcal{F}$ is called strongly 2-universal (S2U).
\end{definition}

Let $a$ be an upper bound on the bit length of signed messages.
In our scheme we employ a $2^{-b+1}$-AS2U family $\mathcal{F}=\{f_\kappa\}_{\kappa\in\{0,1\}^y}$ of functions from $\{0,1\}^a$ to $\{0,1\}^b$, where the tag bit length $b>1$, 
$y={3b+2s}$, and the integer $s$ satisfies the inequality $a \leq (2^s+ 1)(b+s)$.
The explicit construction of the family is presented in Appendix~\ref{sec:app:ASU2-construction}.

The workflow of the distribution stage is the following (see also Fig.~\ref{fig:distribution}).
	
{\it Step 1}. Using a true random number generator, the signer $\mathcal{P}_0$ generates $N^2k$ $y$-bit keys $(\kappa_1, \kappa_2, \ldots ,\kappa_{N^2k})$, where an integer $k$ is the basic parameter of the scheme.
This set is the signing key, and each $\kappa_i$ is an authentication key used later for defining functions from the family $\mathcal{F}$.
Then ${\cal P}_0$ transmits to each $\mathcal{P}_i$ ($i=1,2,\ldots,N$) a subset $$(\kappa_{(i - 1)Nk + 1}, \kappa_{(i - 1)Nk + 2},\ldots,\kappa_{iNk})$$ 
using QKD-assisted OTP encryption.

{\it Step 2}. Each internal recipient $\mathcal{P}_i$ randomly splits the obtained keys into $N$ disjoint ordered subsets of size $k$.
Let $R_{i \rightarrow j}\subset\{(i - 1)Nk,\ldots,iNk\}$ with $i,j\in\{1,\ldots,N\}$ be an ordered subset of $k$ key indices belonging to the chosen $j$th subset of $\mathcal{P}_i$'s keys with respect to the original set generated by the signer.
Then $\mathcal{P}_i$ transmits to every $\mathcal{P}_j$ ($j\neq i, 0<i, \text{and} \ j\leq N$) all the keys belonging to the $j$th subset and the corresponding indices $R_{i \rightarrow j}$.
All the messages at this step are also secured with OTP encryption.
The keys with indices from $R_{i \rightarrow i}$ remain with $\mathcal{P}_i$.

In the end of the distribution stage, each internal recipient possesses a set of $Nk$ authentication keys from the original signing key generated by ${\cal P}_0$:
$k$ keys of this set come directly from the signer and the remaining $(N-1)k$ keys come from the other $N-1$ recipients.
We note that, due to the fact that at step 2 recipients exchange indices $R_{i\rightarrow j}$, each recipient knows the indices of each of their $Nk$ keys with respect to the original set of $N^2k$ keys (signing key) generated by the signer.

\subsection{Signature generation}

In order to generate a signature ${\sf Sig}_m$ for a message $m\in\{0,1\}^a$, the signer applies $N^2k$ functions from the family $\mathcal{F}$, specified by generated authentication keys in the signing key, and obtains $N^2k$ authentication tags:
\begin{equation}
	{\sf Sig}_m := (t_1, \ldots ,t_{N^2k}), \quad t_i := f_{\kappa_i}(m).
\end{equation}

In our scheme, the transmission of a message-signature pair is accompanied by sending a verification level, at which this pair was accepted by the current sender.
We describe the concept of verification levels below, but for now we just require that in order to send the signed message to some internal or external node, 
the signer transmits a triple $(m,{\sf Sig}_m,l_{\max})$, where the positive integer $l_{\max}$ is another basic parameter of our scheme, also to be discussed further.

\subsection{Signature verification}

To describe the verification procedure, we first require the introduction of two important concepts: verification levels and block lists.

\subsubsection{Verification levels}

In contrast to standard computationally secure signature schemes, in the USS scheme, the verification rule is specified by an integer parameter $l\in\{0,1,\ldots,l_{\max}\}$ called verification level.
The idea is that if a message-signature pair is accepted by an honest internal or external recipient at verification level $l\geq 1$, then the security properties of the developed scheme ensure that the same message-signature pair will be accepted (with a probability close to 1) by any other internal or external honest recipient at some verification level $l'\geq l-1$.
At the same time, the scheme is developed in such a way, that if the signer is honest and publishes a message-signature pair $(m,{\sf Sig}_m)$, then no one is able to produce (up to negligible probability) a message-signature pair $(m^\star, \sigma^\star)$ with $m^\star\neq m$ that is accepted by some honest internal or external recipient at verification level $l\geq 0$.

We note that in a forwarding chain ${\cal R}_1\rightarrow {\cal R}_2\rightarrow \cdots \rightarrow {\cal R}_s$ of length $s>2$, where the ${\cal R}_i$s are some internal or external recipients, the fact that ${\cal R}_1$ accepts a message-signature pair $(m,\sigma)$ at the verification level $l\geq 1$ implies that \emph{all} other honest recipients in the chain accept $(m,\sigma)$ at verification levels $\geq l-1$.
This is because if some ${\cal R}_i$ $(2<i \leq s)$ does not accept $(m,\sigma)$ at some level $l'\geq l-1$ from ${\cal R}_{i-1}$, then it appears that ${\cal R}_i$ would not accept $(m,\sigma)$ directly from ${\cal R}_1$, which contradicts our security properties (recall that all recipients in the chain are honest, so the message-signature pair does not change). 
To simplify the analysis we also introduce a special rule that the message-signature pair is forwarded together with the verification levels at which it has been accepted. 

The maximal verification level $l_{\max}>0$ can be varied and belongs to the set of basic parameters of the scheme.
In the next section we formulate exact security statements related to the unforgeability and transferability and also introduce a necessary relation between $N$, $l_{\max}$ and $\omega$.
We also note that if the signer is honest, then the security properties of our scheme imply that the original message-signature pair $(m, {\sf Sig}_m)$ will be accepted by any honest recipient at the maximal verification level $l_{\max}$.
Of course the situation may change if there is an intermediary malicious recipient between the honest signer and the honest recipient.

The lowest value of verification level is 0, and
the fact that a message-signature pair is accepted at zero verification level by some honest recipient does not guarantee its acceptance by other honest recipients.
In order to solve this issue we introduce a majority vote dispute resolution process, also discussed later in the text.

\subsubsection{Block list}

The purpose of the block list is termination of communication between honest nodes and apparently dishonest ones.
It allows preventing exhaustive search forging attacks, and also helps to derive rigorous security statements about upper bounds on a probability of forging attacks.

In the preliminary distribution stage each internal (external) recipient ${\cal P}_i$ (${\cal E}_j$) initialize an empty set ${\sf block\_list}^{\sf int}_i:=\{ \ \}$,   (${\sf block\_list}^{\sf ext}_j:=\{ \ \}$).
These sets are intended to store the labels of blocked nodes.

Each internal recipient ${\cal P}_i$ also initializes a set of counters 
${\sf cnt}_{i,i'}:=0$, where $i'=1,2,\ldots, M$.
These counters are intended to keep a number of failed verification requests in the delegated verification from a particular external recipient ${\cal E}_{i'}$. 

Now we are ready to formalize the verification procedure for internal and external recipients.

\subsubsection{Verification by internal recipients}

Consider an internal recipient ${\cal P}_i$ obtaining from some node ${\cal R}$ a package $(m,\sigma, l_{\rm rec})$, where $m\in\{0,1\}^a$, $\sigma=(t_1,\ldots,t_{N^2k})\in\{0,1\}^{b\times N^2k}$, and $l_{\rm rec}\geq 1$.
To verify the message-signature pair $(m,\sigma)$, the following steps are performed [see also Fig.~\ref{fig:verification}(a)].

\begin{figure}[t]
    \centering
    \includegraphics[width=\linewidth]{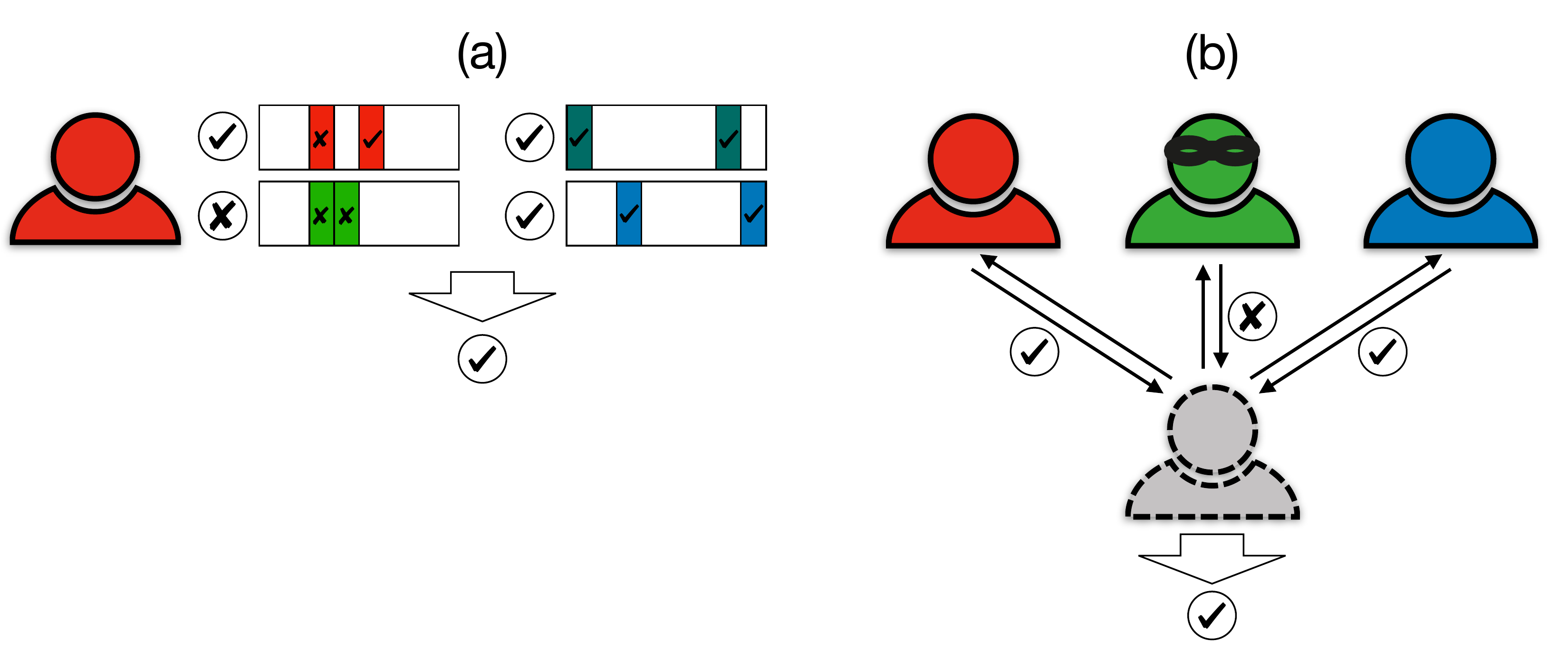}
    \vskip -4mm
    \caption{Signature verification by (a) internal and (b) external recipients.
    Internal recipients first perform tests for each of the possessed signature key blocks [see Eq.~\eqref{eq:T_mijl}] and then make a final decision based on all $N$ test results.
    External recipients send verification requests to $2\omega+1$ of the internal recipients and then make a decision based on the majority principle.}
    \label{fig:verification}
\end{figure}

{\it Step 0}. Internal recipient ${\cal P}_i$ checks whether the sender ${\cal R}\in {\sf block\_list}^{\sf int}_i$.
If it is not the case, ${\cal P}_i$ proceeds with the verification. 
Otherwise, the package is ignored.

{\it Step 1}. Internal recipient $\mathcal{P}_i$ performs $N$ tests corresponding to each of $N$ possessed blocks of $k$ authentication keys:
For each $j\in\{1,\ldots,N\}$ the value
\begin{equation} \label{eq:T_mijl}
	T_{i, j, l}^m = \begin{cases}
	1 &\text{ if } \sum_{r \in R_{j \rightarrow i}} g(f_{\kappa_r}(m), t_r) < s_lk, \\ 
	0 &\text{ otherwise }
    \end{cases} 
\end{equation}
is computed.
Here $g(., .)$ is a comparison function which returns 1 if its arguments are different and 0 otherwise. The value $s_l$ is defined as
\begin{equation}\label{eq:scrit}
	s_l = \left(1-l/l_{\max}\right)s_0, \quad s_0 \in (0,1-2^{1-b}),
\end{equation}
where $s_0$ is the basic parameter of the scheme which denotes the tolerable fraction of incorrect tags for a verification of a subblock corresponding tp a particular recipient at the lowest verification level.

{\it Step 2}.  Internal recipient $\mathcal{P}_i$ obtains a set 
\begin{equation}
    \Lambda = \left\{
        l\in\{0,\ldots,l_{\max}\}: \sum_{j=1}^N T_{i, j, l}^m > T_{l},
    \right\}   
\end{equation}
where the critical number of tests is given by
\begin{equation}\label{eq:Tcrit}
    T_l= \omega + l\omega.
\end{equation}
One can think about $\Lambda$ as a set of verification levels at which the message-signature pair can be accepted.
Note that $\Lambda$ can appear to be empty.
Then the resulting verification level is computed as the maximum in $\Lambda$ or error value --1 for empty $\Lambda$:
\begin{equation}
    l_{{\rm ver},i}(m,\sigma) := 
    \begin{cases}
        \max\Lambda & \text{if }|\Lambda|>0 \\
        -1 & \text{otherwise.}
    \end{cases}
\end{equation}

{\it Step 4}. If $l_{{\rm ver},i}(m,\sigma)\geq l_{\rm rec}-1$, then we say that the message-signature pair is accepted by ${\cal P}_i$ at verification level $l_{{\rm ver},i}(m,\sigma)$.
Otherwise, the message-signature pair is said to be rejected, and the sender ${\cal R}$ is added by ${\cal P}_i$ to the block list: 
${\cal R} \rightarrow {\sf block\_list}^{\sf int}_i$.
If $l_{{\rm ver},i}(m,\sigma)\geq 1$ then ${\cal P}_i$ is allowed to forward the message-signature pair to any other internal or external recipient with a package $(m,\sigma,l_{{\rm ver},i}(m,\sigma))$.

\subsubsection{Delegated verification by external nodes}

Next we describe verification of a proper package $(m,\sigma, l_{\rm rec})$ by an external recipient ${\cal E}_i$.
The following steps are performed [see also~Fig.~\ref{fig:verification} (b)].

{\it Step 0}. External ${\cal E}_i$ first checks whether the sender ${\cal R}\in {\sf block\_list}^{\sf ext}_i$.
If it is not the case, ${\cal E}_i$ proceeds with the verification. 
Otherwise, the package is ignored.

{\it Step 1}.
External recipient ${\cal E}_i$ chooses a subset $\Omega \subset \{1,\ldots, N\}$ such that
${\cal E}_i$ is connected with each ${\cal P}_{i'}$, ${i'}\in \Omega$ by a QKD link.
If the package, obtained by ${\cal E}_i$, comes from some internal node, that is ${\cal R}={\cal P}_j$ for some $j$, then $\Omega$ has to be of size $|\Omega|= 2\omega$ and has to exclude $j$.
Otherwise, if the package is obtained from some other external node, $\Omega$ has to be of size $|\Omega|= 2\omega+1$.
Then ${\cal E}_i$ sends each ${\cal P}_{i'}$, ${i'}\in \Omega$ a request to verify the package $(m,\sigma, l_{\rm rec})$.

{\it Step 2}.
Having received the request, ${\cal P}_{i'}$ first checks whether ${\cal E}_i$ is in the block list ${\sf block\_list}^{\sf int}_{i'}$.
If it the case, then the request is ignored. 
Otherwise, ${\cal P}_{i'}$ runs the verification algorithm described in steps 1--3 from the preceding subsection, forms a response ${\sf resp}_{i'}:=l_{{\rm ver},i'}(m,\sigma)$, and transmits this response back to ${\cal E}_i$.
Moreover, if $l_{{\rm ver},i'}(m,\sigma)<l_{\rm rec}-2$, then ${\cal P}_{i'}$ increments a counter ${\sf cnt}_{i',i}$.
If the counter ${\sf cnt}_{i',i}$ reaches the critical value $M+\omega$, then $i'$ puts ${\cal E}_{i}$ in the block list ${\sf block\_list}^{\sf int}_{i'}$ (the idea behind this operation is to prevent exhaustive-search forgery attacks by using delegated verification requests).

{\it Step 3}.
External recipient ${\cal E}_i$ collects all the responses ${\sf resp}_{i'}$ from ${\cal P}_{i'}, i'\in \Omega$.
If the original package $(m,\sigma,l_{\rm rec})$ has been received from internal node ${\cal P}_j$, then $j$ is added to $\Omega$ and ${\sf resp}_j$ is set to $l_{\rm rec}$.
Then for each $l\in\{-1,0,\ldots,l_{\rm max}\}$ the set
\begin{equation}
    \Omega_l(m,\sigma):=\{i: {\sf resp}_i(m,\sigma) \geq l\}
\end{equation}
is calculated.
The result of the verification is obtained as
\begin{equation}
    l_{{\rm ver},i}^{\rm ext}(m,\sigma):=\max\{l': |\Omega_{l'}(m,\sigma)|\geq \omega+1\}.
\end{equation}

{\it Step 4}. If $l_{\rm ver,i}^{\rm ext}(m,\sigma)\geq l_{\rm rec}-1$ then the message-signature pair $(m,\sigma)$ is said to be accepted at verification level $l_{\rm ver,i}^{\rm ext}(m,\sigma)$ by ${\cal E}_i$ and then ${\cal E}_i$ is allowed to forward the pair in the package $(m,\sigma,l_{\rm ver,i}^{\rm ext}(m,\sigma))$ to any other internal or external recipient.
Otherwise, the message signature-pair is said to be rejected and ${\cal E}_i$ blocks its sender: ${\cal R} \rightarrow {\sf block\_list}^{\sf ext}_i$.

\subsection{Transferability at zero verification level} \label{sec:zerolevel}

The presented verification routines contain an issue related to the minimal verification level $l=0$.
The following situation is possible: The verification protocol run by internal recipient ${\cal P}_i$ with respect to package $(m,\sigma,1)$ results in $l_{{\rm ver},i}(m,\sigma)=0$.
Then ${\cal P}_i$ can be sure that $m$ is produced by ${\cal P}_0$, however it is not guaranteed that $(m,\sigma)$ will be accepted by other honest recipients.

In order to cope with this issue we introduce the majority vote dispute resolution process (or majority vote, for short), also employed in previous USS scheme designs~\cite{Arrazola2016,Amiri2018}.
The majority vote is an expensive (in terms of communication costs) routine and is not a necessary part of the USS scheme workflow.
However, its potential possibility is necessary for providing security of the scheme. 

We note that in the previous USS scheme designs~\cite{Arrazola2016,Amiri2018}, an additional --1th verification level is reserved for the purposes of the majority vote process.
The security of the process is justified by establishing protection against non-transferability attacks at all verification levels from $l_{\max}$ down to --1.
Though formally all the claimed security statements in~\cite{Arrazola2016,Amiri2018} are completely correct, there is as issue with the fact that unforgeability is provided for verification levels from $l_{\max}$ down to 0.
This fact opens a security loophole related to the possibility of a malicious recipient to forge a message-signature pair $(m^\star,\sigma^\star)$ that is acceptable by honest recipients at the --1th verification level and then initiate the majority vote process.
Then honest recipients will accept $(m^\star,\sigma^\star)$, although $m^\star$ was not signed by the legitimate signer.
The straightforward solution to this issue is to extend unforgeability down to $l=-1$ as well or to increase $l_{\max}$ and consider the majority vote process at the zero verification level.
We choose the second solution in our work.

The operation of the majority vote employs an unconditionally secure broadcast protocol run within the set of internal recipients.
For the description of the broadcast protocol we refer the reader to the seminal works in ~\cite{Pease1980, Lamport1982}.
Here we only recall that the broadcast protocol allows a node in the network to transmit a message to a number of other nodes in such a way that it is guaranteed that all the honest recipients will obtain this message, and also it is guaranteed that if an honest recipient obtains a message from the broadcast protocol, then that recipient can be sure that other honest recipients obtained the same message.
As shown in Refs.~\cite{Pease1980, Lamport1982}, in the presence of no more than $\omega$ dishonest nodes, the protocol requires $\omega+1$ rounds of communication, at which nodes transmit messages through unconditionally secure authentication channels between them. 
Moreover, the broadcast protocol for $N$ parties can be realized only if $\omega<N/3$, which is the case in our setup.

\subsubsection{Majority vote process within the set of internal recipients}

Here we describe the majority vote process.
It is allowed to be launched by any internal recipient ${\cal P}_i$ for a message-signature pair $(m,\sigma)$ only in the case of $l_{{\rm ver},i}(m,\sigma)=0$.

{\it Step 1.}
The initiator of a majority vote process ${\cal P}_i$ broadcasts the message-signature pair $(m,\sigma)$ to all other internal recipients.

{\it Step 2.}
Every internal recipient ${\cal P}_j$ $(j=1,2,\ldots,N)$ broadcasts the result of its verification in the form
${\sf vote}_j(m,\sigma) = l_{{\rm ver},j}(m,\sigma)$.

{\it Step 3.}
Every node computes the result of the majority vote protocol in the form
\begin{equation}
    {\sf MV}(m,\sigma):=\begin{cases}
        \checkmark & \text{if }\sum_{j=1}^{N}\widetilde{\sf vote}_j>N/2 \\
        \oslash & \text{otherwise},
    \end{cases}
\end{equation}
where 
\begin{equation}
    \widetilde{\sf vote}_j := \begin{cases}
        1 & \text{if }{\sf vote}_j\geq 0 \\
        0 & \text{otherwise},
    \end{cases}
\end{equation}
and it is assumed that the vote of the initiator of the process ${\sf vote}_i=0$.
We note, that due to the properties of the broadcast protocol, all the honest internal nodes obtain the same value of ${\sf MV}(m,\sigma)$.
If the process results in ${\sf MV}(m,\sigma)=\checkmark$ $(\oslash)$, then $(m,\sigma)$ is said to be accepted (rejected) by the majority vote.

One can also see that if  there exists a set
$\Omega \subset \{1,2,\ldots,N\}$, such that $|\Omega| \geq \omega+1$ and for every $i\in \Omega$, ${\sf vote}_i\geq 2$, then all honest nodes can conclude that the originator of the majority vote protocol is dishonest (with up to negligible probability of a fail).
This follows from the fact that there is at least one honest recipient ${\cal P}_i$ with $i\in\Omega$, so the originator of the majority vote should accept $(m,\sigma)$ at verification level $l\geq 1$.
So the rules of the USS operation are supplemented by punishment for dishonest conduct.

\subsubsection{Majority vote results verification by external recipients}

Here we describe how an external recipient ${\cal E}_i$ can obtain the results of the majority vote performed within the set of internal recipients with respect to some message-signature pair $(m,\sigma)$.

{\rm Step 1.} 
Node ${\cal E}_i$ chooses a subset $\Omega \subset \{1,\ldots,N\}$ of size $|\Omega|= 2\omega+1$, such that ${\cal E}_i$ is connected with every node ${\cal P}_{i'}$ with $i'\in\Omega$ by a QKD link.
${\cal E}_i$ sends a majority vote verification request consisting of $(m,\sigma)$ to every internal recipient from the set $\Omega$.

{\rm Step 2.} 
Having received the request, ${\cal P}_{j}$ makes a response ${\sf MV\_resp}_{j}:={\sf MV}(m,\sigma)$ if there was a majority vote with respect to $(m,\sigma)$, or ${\sf MV\_resp}_{j}:=\bot$ otherwise.

{\rm Step 3.}
Node ${\cal E}_i$ collects all the responses $\{{\sf MV\_resp}_{j}\}_{j\in\Omega}$.
Let $\#[\checkmark]$ and $\#[\oslash]$ be number of occurrences of responses $\checkmark$ and $\oslash$ in $\{{\sf MV\_resp}_{j}\}_{j\in\Omega}$, respectively.
The result of delegated verification of the majority vote results by the external node ${\cal E}_i$ is given by
\begin{equation}
    {\sf MV}^{{\sf ext}}_{i}(m,\sigma) := \begin{cases}
        \checkmark &\text{if } \#[\checkmark]\geq \omega+1\\
        \oslash &\text{if } \#[\oslash]\geq \omega+1\\
        \bot &\text{otherwise.}
    \end{cases}
\end{equation}

If ${\sf MV}^{{\sf ext}}_{i}(m,\sigma)=\checkmark (\oslash)$ then we say that $(m,\sigma)$ is said to be accepted (rejected) by ${\cal E}_i$ within majority vote results verification.
Here ${\sf MV}^{{\sf ext}}_{i}(m,\sigma)=\bot$ means there was no majority vote with respect to $(m,\sigma)$ in the internal network.
The security properties of the developed scheme ensure that
if ${\sf MV}^{{\sf ext}}_{i}(m,\sigma)=\checkmark$ for some honest external recipient ${\cal E}_i$, then for every other honest external recipient ${\cal E}_j$ the described majority vote verification protocol will result in ${\sf MV}^{{\sf ext}}_{j}(m,\sigma)=\checkmark$; in addition, for any honest internal recipient it is true that ${\sf MV}(m,\sigma)=\checkmark$.

We note that an extra rule can be added that some particular external nodes can insist on running the majority vote with respect to some message-signature pair.
Then after requesting no more than $\omega+1$ nodes the majority vote will happen  ($\omega$ nodes can be dishonest and deny the start of voting) and each of the external nodes will be able to receive its result by using the described protocol.

\section{Security analysis}\label{sec:security}

In this section we introduce security definitions and corresponding security statements (all proofs are placed in Appendix~\ref{sec:app:proofs}).
Here we also demonstrate how the security conditions impose dependences between the basic parameters of the scheme.

\subsection{Signature acceptability} \label{sec:acceptability}

We start with a natural way to demand that all honest (internal and external) recipients have to accept a message-signature pair $(m,{\sf Sig}_m)$ generated by the honest signer.
The nontriviality of this condition for the USS scheme comes from the fact that the verification key of ${\cal P}_i$ contains authentication keys that come from all other, including possibly malicious, recipients.
These dishonest recipients can try to foil the verification procedure performed by ${\cal P}_i$ with respect to $(m,{\sf Sig}_m)$ by transferring ‘rubbish’ keys at the second step of the distribution stage [see, e.g., ${\cal P}_2$ in Fig.~\ref{fig:distribution}(b)].
In the original design of the USS scheme in Ref.~\cite{Amiri2018} an acceptance of $(m,{\sf Sig}_m)$ by any honest recipient at zero verification level.
However, the acceptance of a message-signature pair $(m,\sigma)$ exactly at the zero verification level (but not at higher levels) closes off the possibility of reliable forwarding of $(m,\sigma)$ by ${\cal P}_i$ to other recipients without appealing to the majority vote process.
The problem here is that ${\cal P}_i$ is not able to distinguish between the two following situations.
The first one is where ${\cal P}_0$ is honest and an acceptance of the message-signature pair only at the lowest verification level is due to an attack of the malicious recipient coalition [but $(m,\sigma)$ will in fact be accepted by other honest recipients].
The second situation is where ${\cal P}_0$ is malicious and tries to perform a non-transferability attack (to be discussed further).
In our design we use the following definition.

\begin{definition}[signature acceptability].
    We say that the USS scheme provides a signature acceptability if a message-signature pair generated by an honest signer is accepted by any honest internal or external recipient at the maximal verification level $l_{\max}$; that is, for any $m\in\{0,1\}^a$, a verification procedure run by an internal (external) recipient ${\cal P}_i$ (${\cal E}_j$) with respect to the package $(m,{\sf Sig}_m,l_{\max})$ results in $l_{{\rm ver},i}(m,{\sf Sig}_m)=l_{\max}$ [$l_{{\rm ver},j}^{\rm ext}(m,{\sf Sig}_m)=l_{\max}$], assuming that the sender of this package is not in the block list of ${\cal P}_i$ (${\cal E}_j$).
\end{definition}
The signature acceptability property guarantees that malicious recipients are not able to decrease transferability by cheating during the distribution stage.
This condition leads to the appearance of an upper bound on a number of dishonest participants $\omega$ depending on a maximal verification level $l_{\max}$.

\begin{theorem}[upper bound on $\omega$]. \label{thm:trade-off}
	The USS scheme, described in the text, provides  the signature acceptability if and only if
	\begin{equation} \label{eq:trade-off}
	    \omega < \frac{N}{2+l_{\max}}.
	\end{equation}
\end{theorem}
We design our scheme to possess the signature acceptability, so we assume that~\eqref{eq:trade-off} is fulfilled.
We note that from~\eqref{eq:trade-off} it follows that $\omega<N/3$, since $l_{\max}$ is a positive integer.
Moreover, to have $\omega=1$ it necessary to have $N\geq 4$.

\subsection{Signature unforgeability}

The second security property relates to the assurance that no one, other than the signer, can generate a valid signature.
To describe this condition we firstly introduce a formal definition of a signature forgery.

\begin{definition}[signature forgery].
	Consider a situation where there is a coalition of dishonest recipients.
	Let the coalition possess a valid message-signature pair $(m, {\sf Sig}_m)$. 
	Suppose, that for some message $m^\star \neq m$ the coalition generates a guess for a corresponding signature $\sigma^\star$. 
	If the pair $(m^\star, \sigma^\star)$ is accepted by at least one honest internal (external) recipient $\mathcal{P}_i$ ($\mathcal{E}_j$) at some verification level $l\geq 0$, then we say that a forgery event happened.
\end{definition}

We note that the definition of a forgery introduced here is an extension of that used in Ref.~\cite{Amiri2018}.
In particular, it includes consideration of the lowest verification level employed in the majority vote process.
The motivation of such an extension of the forgery definition comes from the security issues described in Sec.~\ref{sec:zerolevel}.

The following theorem states that in the considered USS scheme the probability of forgery drops exponentially with the value of $k$.

\begin{theorem}[probability of forgery]. \label{thm:forgery}
	For the USS scheme described herein, the upper bounds on a forgery event hold,
	\begin{equation} \label{eq:forg-bound-1}
		\Pr[{\rm forgery}] < J({N,M,\omega}) 2^{-k(b - 1)[1 - s_0 - H_2(s_0)/(b - 1)]}
		%\Pr[{\rm Forgery}] < \frac{4}{27}N^3 \mu 2^{-k(b - 1)\left(1 - s_0 - \frac{H_2(s_0)}{b - 1}\right)}
	\end{equation}
	for $s_0<\frac{1}{2}$ and 
	\begin{equation} \label{eq:forg-bound-2}
		\Pr[{\rm forgery}] < 
		J({N,M,\omega}) e^{-2k(1 - s_0 - 2^{1-b})^2}
		%\frac{4}{27}N^3\mu  e^{-2k(1 - s_0 - 2^{1-b})^2}
	\end{equation}
	for $s_0 < 1 - {2^{1-b}}$,
	where 
	\begin{equation}
	    J({N,M,\omega}):=N^2[\omega+M(\omega+M)]
	\end{equation}
	and $H_2(\cdot)$ is a standard Shannon binary entropy.
\end{theorem}

We note that although the first bound~\eqref{eq:forg-bound-1} is tighter, the second bound~\eqref{eq:forg-bound-2} appears to be more practical in consideration of both unforgeability and transferability conditions.

\subsection{Message transferability} \label{sec:transferability}

The third condition relates to the requirement that if an honest (internal or external) recipient accepts a message-signature pair at the verification level $l\geq 1$, then another honest recipient will accept the same pair at least at the $(l-1)$th verification level.
A crucial difference compared to the previous condition is that the sender may belong to the coalition of dishonest participants.

\begin{definition}[non-transferability].
    Consider a situation where there is a coalition of dishonest nodes (${\cal P}_0$ may or may not belong to this coalition).
	Let the coalition output a message-signature pair $(m, \sigma)$.
	We say that a non-transferability event happens if some honest internal or external recipient accepts $(m, \sigma)$ at the verification level $l\geq 1$, but other internal or external recipient do not accept $(m, \sigma)$ at any verification level $l'\geq l-1$.
\end{definition}

We note that, in contrast to previous USS scheme designs~\cite{Arrazola2016,Amiri2018}, we also include the possibility that the malicious coalition does not include the signer ${\cal P}_0$.
Such a coalition may try to corrupt the valid signature ${\sf Sig}_m \rightarrow \sigma$ in such a way that $(m,\sigma)$ is accepted by ${\cal P}_i$ at the verification level $l\geq 1$, but is rejected by ${\cal P}_j$ at the verification level $l'=l-1$.

The next theorem states that the probability of a non-transferability event also drops exponentially with the value of $k$.

\begin{theorem}[probability of non-transferability]. \label{thm:non-transferability}
    For the USS scheme, described herein, the upper bound on the probability of a non-transferability event holds true,
    \begin{equation}\label{eq:non-trans-bound}
        \Pr[\text{non-transferability}] \leq 2N^2(N-1)
            e^{-k\Delta s^2/2},
    \end{equation}
    where $\Delta s=s_0/l_{\max}$.
\end{theorem}

Comparing Eqs.~\eqref{eq:forg-bound-1}, \eqref{eq:forg-bound-2}, and \eqref{eq:non-trans-bound}, one can see that $s_0$ affects unforgeability and non-transferability bounds in opposite ways: Increasing $s_0$ improves the non-transferability bound but 
weakens the unforgeability one.
Thus, we arrive at the necessity of optimizing $s_0$ with respect to the desired security parameters in practical realizations of the scheme.

\subsection{Non-repudiation} \label{sec:repudiation}

The next security property we consider is non-repudiation.
It states that the signer is not able to refuse authorship of a signed message.
One can see that this condition closely relates to the transferability.
In line with the design of the USS scheme in Ref.~\cite{Amiri2018}, we consider the repudiation issue in the context of the majority vote process.

\begin{definition}[repudiation].
	Suppose that a coalition of dishonest nodes (${\cal P}_0$ may or may not belong to this coalition) outputs a message-signature pair $(m, \sigma)$.
	We say that a repudiation event happens if some honest (internal or external) recipient accepts $(m,\sigma)$ at the verification level $l\geq 1$, but the majority vote results in ${\sf MV}(m,\sigma)=\oslash$.
\end{definition}

We state that the probability of a repudiation event can be upper bounded by the same expression as the non-transferability.

\begin{theorem}[probability of repudiation]. \label{thm:reputication}
For the USS scheme, described in the text, the following holds true:
    \begin{equation}\label{eq:repudiation-bound}
        \Pr[\text{repudiation}]\leq \Pr[\text{non-transferability}].
    \end{equation}
\end{theorem}

We conclude this subsection by stating the relation between the results of the majority vote process performed by internal recipients and the majority vote result verification by external recipients.

\begin{theorem}[proper operation of the majority vote results verification]. \label{thm:majority_vote}
    If the majority vote result verification with respect to the message-signature pair $(m,\sigma)$ is run by an external recipient ${\cal E}_i$ in the absence of a corresponding majority vote run by internal recipients, then it outputs ${\sf MV}^{\sf ext}_{m,\sigma}=\bot$.
    Otherwise, ${\sf MV}^{\sf ext}_i(m,\sigma)={\sf MV}(m,\sigma)$.
\end{theorem}

\subsection{Security of block lists operation}

The final security statement is related to a block lists operation.

\begin{definition}[false blocking].
    We say that a false blocking event happens if some honest internal or external recipient appears on a block list of some other honest internal or external recipient.
\end{definition}

It appears that the probability of this undesirable event can be upper bounded by the probability of forgery and non-transferability.

\begin{theorem}[proper block lists operation]. \label{thm:block-lists}
    For the USS scheme, described herein, the following upper bound on the probability of a false blocking event holds true:
    \begin{multline} \label{eq:false-blocking}
        \Pr[\text{false blocking}] \\ \leq \Pr[\text{forgery}]+\Pr[\text{non-transferability}].
    \end{multline}
\end{theorem}

Thus, according to Theorems~\ref{thm:forgery} and~\ref{thm:non-transferability} the probability of false blocking is upper bounded by a function decreasing exponentially with $k$.

\section{Performance analysis}\label{sec:performance}

Here we discuss practical aspects of implementing the considered QKD-based USS scheme.
We are particularly interested in the consumption of symmetric keys, generated within the internal QKD subnetwork and used for providing OTP encryption at the preliminary distribution stage.

Let us introduce a security parameter $\varepsilon_{\rm tot}$ that bounds probabilities of successful forgery and non-transferability events as follows:
\begin{equation} \label{eq:epsilon-tot}
	\begin{aligned}
	    &\Pr[\text{forgery}] \leq \frac{\varepsilon_{\rm tot}}{2},\\
	    &\Pr[\text{non-transferability}] \leq \frac{\varepsilon_{\rm tot}}{2}.
	\end{aligned}
\end{equation}
From the practical point of view, it is reasonable to fix the value of $\varepsilon_{\rm tot}$ 
at the level of the QKD security parameter, which is commonly of the order of $10^{-9} -- 10^{-12}$~\cite{Kiktenko2016}.
%The most important resource for the realization of the scheme are symmetric key generated within the QKD network.

There are two basic types of links within the internal subnetwork: 
(i) links between the signer and  internal recipients and 
(ii) links between internal recipients.
We denote these types by sr and rr, respectively.
The key consumption for each type of links in the distribution stage is given by
\begin{equation}\label{eq:keyconsumptions}
	L_{{\rm sr}} = Nky, \quad L_{{\rm rr}} = 2k(y+\lceil \log_2Nk \rceil),
\end{equation}
where, again $N$ is the number of internal recipients, $y$ is the length of a key defining an element from the employed AS2U family, and $k$ is the number of single block authentication keys appearing in decaying exponents of Eqs.~\eqref{eq:forg-bound-1}, \eqref{eq:forg-bound-2}, and \eqref{eq:non-trans-bound}.
We also note that $y$ scales logarithmically with the maximal message length $a$.
The total key consumption for all links in the internal subnetwork can be calculated as follows:
\begin{equation} \label{eq:keyconsumptiontot}
	L_{\rm tot} = NL_{{\rm sr}} + \frac{N(N-1)}{2}L_{{\rm rr}}.
\end{equation}
Substituting Eq.~\eqref{eq:keyconsumptions} into Eq.~\eqref{eq:keyconsumptiontot}, we see that the total key consumption generally scales as $N^2$ with the growth of internal subnetwork size $N$.
The dependence on the size $M$ of the external subnetwork appears in the upper bound on the forgery event and is logarithmic.

The lengths of required symmetric keys, given by $L_{\rm sr}$ and $L_{\rm rr}$, limit the rate ${\sf rate}_{\rm USS}^{\max}$ at which sets of signing and verification keys can be generated.
This rate can be calculated as
\begin{equation}\label{eq:rate-uss}
    {\sf rate}_{\rm USS}^{\max}=\min\left(\frac{{\sf rate}_{\rm sr}^{\min}}{L_{\rm sr}},
    \frac{{\sf rate}_{\rm rr}^{\min}}{L_{\rm rr}}\right),
\end{equation}
where
\begin{equation}
    {\sf rate}_{\rm sr}^{\min}:=\min\limits_{i>0} ({\sf rate}_{0i}), \quad
    {\sf rate}_{\rm rr}^{\min}:=\min\limits_{\begin{smallmatrix}
    i\neq j \\
    i,j>0
    \end{smallmatrix}} ({\sf rate}_{ij})
\end{equation},
and ${\sf rate}_{ij}$ is the secret key generation rate of a QKD link connecting ${\cal P}_i$ and ${\cal P}_j$ (recall that ${\cal P}_0$ is the signer in our scheme).
Assuming that all the nodes are connected with same QKD devices and a transmittance between the nodes is determined only by the distance between them, the secret key generation rates can be approximated as
\begin{equation} \label{eq:rate-for-distance}
    {\sf rate}_{ij}={\sf rate}^{(0)}\eta({\sf dist}_{ij})={\sf rate}^{(0)}e^{-\gamma {\sf dist}_{ij}},
\end{equation}
where ${\sf rate}^{(0)}$ is the secret key generation rate at zero distance, $\eta({\sf dist}_{ij})=\exp(-\gamma {\sf dist}_{ij})$ is the transmittance given as a function of the distance ${\sf dist}_{ij}$ between ${\cal P}_i$ and ${\cal P}_j$, and $\gamma$ is a loss coefficient (e.g. for a standard optical fiber it corresponds to a value of 0.2 dB/km).
Substituting~\eqref{eq:rate-for-distance} into~\eqref{eq:rate-uss} and approximating $L_{\rm rr}$ as $2ky$, we obtain
\begin{equation}
    {\sf rate}_{\rm USS}^{\max} \approx \frac{{\sf rate}^{(0)}}{2Nky}\min(2e^{-\gamma{\sf dist}^{\max}_{\rm sr}}, Ne^{-\gamma{\sf dist}^{\max}_{\rm rr}}),
\end{equation}
where
\begin{equation}\label{eq:total-rate-from-dist}
     {\sf dist}^{\max}_{\rm sr}:=\max\limits_{i>0} ({\sf dist}_{0i}), \quad
    {\sf dist}_{\rm rr}^{\max}:=\max\limits_{\begin{smallmatrix}
    i\neq j \\
    i,j>0
    \end{smallmatrix}} ({\sf dist}_{ij}).
\end{equation}
One can expect from~\eqref{eq:total-rate-from-dist} that for large $N$, the rate of key generation for the considered USS scheme is limited by the maximal distance between the signer and recipients ${\sf dist}^{\max}_{\rm sr}$.

In order to compute the key consumption, we solve numerically the constraint optimization problem of minimizing $L_{\rm tot}$ with respect to the value of $k$, tag length $b$, and tolerable fraction of incorrect tags $s_0$, 
keeping fulfilment of the inequalities~\eqref{eq:epsilon-tot} for prefixed values of message length $a$, security parameter $\varepsilon_{\rm tot}$, 
number of internal recipients $N$, number of internal recipients $M$, maximal number of dishonest nodes in the internal subnetwork $\omega$, and maximal verification level $l_{\max}$.
We perform the optimization as follows.
We fix the value of $b$, and then by using the inequalities~\eqref{eq:epsilon-tot} and bounds~\eqref{eq:forg-bound-1},~\eqref{eq:forg-bound-2}, and ~\eqref{eq:non-trans-bound} obtain appropriate values of $k$ and $s_0$.
More concretely, we find $s_0$ and $k$ such that bounds on forgery and non-transferability events are almost the same and the sum of bounds is approximately $\varepsilon_{\rm tot}$, in accordance with~\eqref{eq:epsilon-tot}.
Then we calculate the corresponding key consumption using Eqs.~\eqref{eq:keyconsumptions} and~\eqref{eq:keyconsumptiontot}.
The above procedure is repeated for values of $b$ taken from a given range $b$ (for our purpose we considered $b\in\{2,3,\ldots,20\}$), 
and the value of $b$, together with the corresponding values of $k$ and $s_0$, providing the minimal total key consumption $L_{\rm tot}$ is chosen.
We provide full details of the optimization procedure for finding $k$, $s_0$, and $b$ in Appendix~\ref{sec:app:optimization}.

To demonstrate the results, we consider two regimes of the USS scheme operation that corresponds to extreme cases of the trade-off between $l_{\max}$ and $\omega$ given by Eq.~\eqref{eq:trade-off}. 
The first one, which we call the minimal transferability regime, is characterized by the minimal nontrivial value of $l_{\max}=1$ and the maximal possible value of $\omega=\lceil N/3\rceil-1$.
The second one, which we call the maximal transferability regime, is characterized by $l_{\max} = N-3$, and $\omega=1$.

First, we present the resulting key consumption as a function of $N$ for the case of $a=8$ Mbits, $\varepsilon_{\rm tot}=10^{-10}$, and $M=5$ in Fig.~\ref{fig:key-consumption}.
To demonstrate the importance of additional optimization with respect to the authentication tag length $b$, we also show key consumption for the minimal tag length $b=2$, considered in the seminal paper in ~\cite{Amiri2018}.

One can see that the key consumption is strongly affected by the value of $l_{\max}$, and the maximal transferability regime appears to be the most costly.
The reason for this is the fact that the value of $l_{\max}$ drastically affects the prefactor of $k$ in the decaying exponent in the non-transferability bound~\eqref{eq:non-trans-bound}.
More precisely, to keep the same upper bound on $\Pr[\text{non-transferability}]$ while increasing $l_{\max}$ from $l_{\max}=l_1$ to some $l_{\max}=l_2$, one has to provide an increase of $k$ from $k=k_1$, which corresponded to $l_{\max}=l_1$, up to
\begin{equation}
    k=\left(\frac{l_2}{l_1}\right)^2k_1.
\end{equation}
In the maximal transferability regime  $l_{\max}$ grows with $N$ which results in the significant increase of $L_{\rm sr}$ and $L_{\rm rr}$ due to increase of $k$.

\begin{figure}[]
   \centering
    \includegraphics[width=\linewidth]{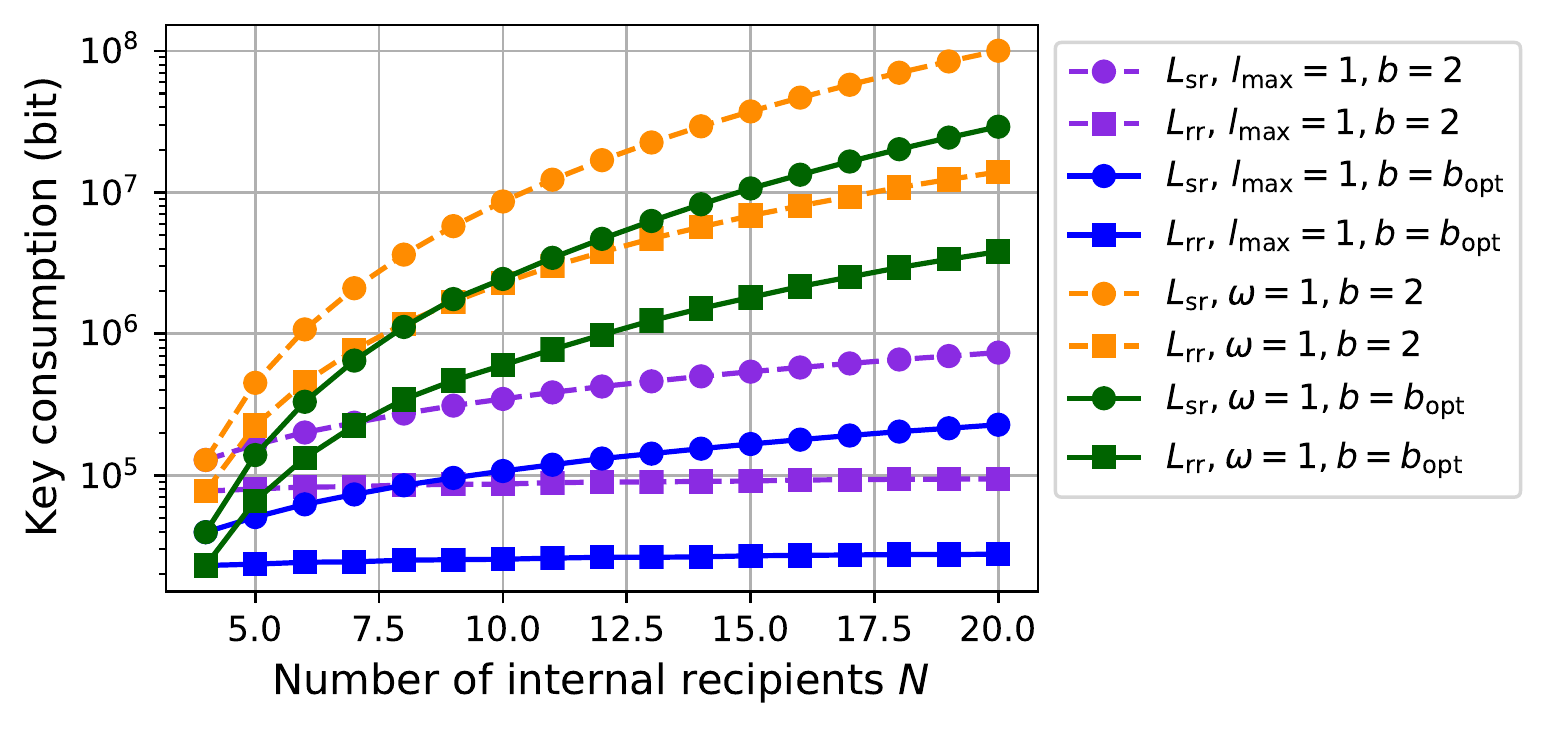}
    \vskip 2mm
    \caption{Optimized secret key consumption required for OTP encryption at the preliminary distribution stage as a function of the number of internal recipients $N$ for different values of $l_{\rm max}$ and $\omega$ satisfying trade-off~\eqref{eq:trade-off}. 
    The signed message is of length $a=8$ Mbits, the security parameter $\varepsilon_{\rm tot}=10^{-10}$, and the number of external recipients $M=5$.
    The results for the fixed ($b=2$) and the optimized ($b=b_{\rm opt}$) tag length are shown.}

    \label{fig:key-consumption}
\end{figure}

To provide a more detailed picture, we also present the results of optimization for some particular configurations of the USS scheme in Table~\ref{tab:rslts}.
We show the results both for the case full optimization, where $k$, $s_0$, and $b$ are optimized, and for the case where $k$ and $s_0$ are optimized with respect to the fixed value of $b=2$.
One can see that even in the worst  case considered, the key consumption is of the order of 10 Mbits per link.
Also note that increasing the message length from 8 Mbits to 32 Mbits, as well as increasing the security level from $\varepsilon_{\rm tot}=10^{-10}$ to $\varepsilon_{\rm tot}=10^{-12}$, and increasing the number of external recipients from $M=10$ to $M=100$ has a mild effect on key consumption.
Given the fact that modern QKD devices demonstrate the capacity of key generation of several Mbits per second~\cite{Yuan2018}, 
we can conclude that the developed QKD-assisted USS scheme appears to be suitable for signing one message per about 10 s in moderate QKD networks consisting of about ten nodes.

\begin{table*}[ht]
	\centering
	\begin{tabular}{|c|c|c|c|c|c||c|c|c|c|c|c||c|c|c|c|c|}
    \hline
    \multicolumn{6}{|c||}{Input parameters} & \multicolumn{6}{c||}{Results of optimization for $b=b_{\rm opt}$} & \multicolumn{5}{c|}{Results of optimization for $b=2$}
    \\ \hline
        $N$ & $M$ & $\omega$ & $l_{\max}$ & $a$ & $\varepsilon_{\rm tot}$ & $k$ & $b$ & $s_0$ & $L_{{\rm sr}}$ & $L_{{\rm rr}}$ & ${\sf sig\_len}$ & $k$ & $s_0$ & $L_{{\rm sr}}$ & $L_{{\rm rr}}$ & ${\sf sig\_len}$\\ \hline
        4 & 0 & 1 & 1 & 8 Mbits & $10^{-10}$ & 125 & 7 & 0.658 & 37.6 kbits & 21.5 kbits & 151 kbits &
        482 & 0.334 & 121 kbits & 72.5 kbits & 482 kbits \\
        4 & 10 & 1 & 1 & 8 Mbits & $10^{-10}$ & 136 & 6 & 0.630 & 39.3 kbits & 22.8 kbits & 157 kbits & 
        510 & 0.325 & 128 kbits & 76.7 kbits & 510 kbits \\ %\hline
        10 & 10 & 1 & 7 & 8 Mbits & $10^{-10}$ & 2947 & 9 & 0.996 & 2.33 Mbits & 587 kbits & 23.3 Mbits & 
        13517 & 0.465 & 8.25 Mbits & 2.19 Mbits & 82.5 Mbits \\ %\hline
        10 & 10 & 3 & 1 & 8 Mbits & $10^{-10}$ & 147 & 6 & 0.632 & 106 kbits & 25.3 kbits & 1062 kbits & 
        549 & 0.326 & 343 kbits & 85.8 kbits & 3.35 Mbits \\
        10 & 10 & 2 & 2 & 8 Mbits & $10^{-10}$ & 403 & 6 & 0.766 & 291 kbits & 70.8 kbits & 2.84 Mbits & 
        1511 & 0.395 & 944 kbits & 242 kbits & 9.22 Mbits \\ 
        10 & 10  & 2 & 2 & 32 Mbits & $10^{-10}$ & 403 & 6 & 0.766 & 307 kbits & 74.0 kbits & 3.00 Mbits & 
        1511 & 0.395 & 974 kbits & 248 kbits & 9.51 Mbits \\ 
        10 & 10 & 2 & 2 & 8 Mbits & $10^{-12}$ & 475 & 6 & 0.758 & 343 kbits & 83.5 kbits & 3.35 Mbits & 1783 & 0.391 & 1.09 Mbits & 285 kbits & 10.8 Mbits \\ 
        10 & 100 & 2 & 2 & 8 Mbits & $10^{-10}$ & 414 & 6 & 0.756 & 299.2 kbits & 72.8 kbits & 2.92 Mbits & 1552 & 0.39 & 970 kbits & 249 kbits & 9.47 Mbits\\ 
        \hline
	\end{tabular}
	\caption{Results of the numerical optimization for different configurations of the developed QKD-assisted USS scheme.
	The results for the fixed $(b=2)$ and the optimized $(b=b_{\rm opt})$ tag length are shown.
	Here $L_{\rm sr}$ and $L_{\rm rr}$ are lengths of secret key required to be distributed between the signer and each of the internal recipients and between each pair of recipients, respectively, and ${\sf sig\_len}$ is the resulting signature length.}
	\label{tab:rslts}
\end{table*}

Finally, we recall that the workflow of the considered scheme is based on using perfect authenticated channels between all the parties in the network, and establishing these channels implies symmetric key consumption as well.
Recent progress in the development of lightweight unconditionally secure authentication schemes~\cite{Kiktenko2020} shows that using a key recycling technique~\cite{Portmann2014} allows decreasing key consumption to provide an $(
\varepsilon_{\rm QKD}+\varepsilon_{\rm auth})$-secure authentic channel down to 
\begin{equation}
    L_{\rm auth}=\lfloor-\log_2 \varepsilon_{\rm auth} \rfloor+1 
\end{equation}
bits per message, where 
$\varepsilon_{\rm QKD}$ is the security level of the employed symmetric key obtained with QKD (see Ref.~\cite{Kiktenko2020} for more details).
So, even considering $\varepsilon_{\rm auth}$ several orders smaller than $\varepsilon_{\rm tot}$, we have the key consumption for the single authenticated channel to be of the order of tens of bits, which is practically negligible compared to the consumption in the main USS scheme (e.g. in the case of $\varepsilon_{\rm auth}=10^{-14}$ one has $L_{\rm auth}=47$ bits).

Note that $L_{\rm auth}$ ($2L_{\rm auth}$) bits of symmetric keys from sr (rr) links at the distribution stage and then the $L_{\rm auth}$ key from each link on the route of a message-signature transfer through the global network are required.
The broadcast channels in the majority vote process require no more than $(\omega+1)L_{\rm auth}$ bits of  symmetric keys from rr links, and $N$ broadcast channels are required in total.
Thus, e.g., for $N=10$ and $\varepsilon_{\rm auth}=10^{-14}$ the resulting key consumption from a link does not exceed several kbits, which is much less than the key consumption for providing OTPs at the preliminary distribution stage.
We note that the additional optimization of the key consumption via paralleling of broadcast channels is also possible.

\section{Conclusion and outlook}\label{sec:conclusion}

In the present work, we have developed a universal hashing-based QKD-assisted multiparty USS scheme.
The scheme operates in a QKD network consisting of two subnetworks: a moderately trusted internal one, where the number of malicious nodes is upper bounded by a threshold $\omega$, and an untrusted external one, where the number of malicious nodes is unbounded.
The signer belongs to the internal subnetwork, while the generated message-signature pair can be securely forwarded through the whole network.
The absence of a trust assumption with respect to the external subnetwork is compensated by (i) stronger requirements on the connectivity in which the nodes in the internal subnetwork have to be connected in an all-to-all fashion, while each external recipient has to be connected only with $2\omega+1$ internal recipients; (ii) much higher symmetric key consumption from QKD-links within the internal subnetwork; and (iii) the necessity of the internal recipients' assistance in the verification process run by the external recipient.
The secret key consumption has logarithmic growth with a maximal signed message length that makes the scheme suitable for practical use.

We also have conducted a security analysis of the scheme and adjusted the workflow of the scheme to prevent the possibility of the adversary (a coalition of adversaries) decreasing transferability of messages, conducting forgery, and non-repudiation attacks.

We have performed numerical optimization of the developed scheme parameters to minimize the secret key consumption. 
The results of the optimization show that the key consumption level for networks of about ten nodes is compatible with the capabilities of contemporary QKD devices.
We hope that the obtained results will bring us closer to the deployment of USS schemes in real QKD networks.

As the main shortcoming of the protocol we note the restriction on a number of malicious nodes in the internal subnetwork: 
In order to tolerate $\omega$ malicious nodes in the case of the minimal transferability level $l_{\max}=1$, it is necessary to have $N>3\omega$ nodes in the internal subnetwork.
It is noteworthy that this bound coincides with the one for the unconditionally secure Byzantine agreement protocol~\cite{Pease1980, Lamport1982}.
An important open question is whether this bound can be improved for the  type of QKD-assisted USS schemes considered.

We also note that an interesting direction for further study is consideration of the developed scheme in the framework of unconditionally secure distributed ledgers~\cite{Kiktenko2018,Fedorov2018}. 
In particular, the need for the assistance from internal subnetwork recipients for the verification of message-signature pairs by external subnetwork recipients resembles the idea of the proof of an authority consensus mechanism in blockchains. 
In this way, the consideration of the QKD-assisted USS scheme for an unconditionally secure consensus protocol is one of the potential avenues for future research.

\section*{Acknowledgments}
We thank P. Wallden for fruitful discussions and useful comments. 
This work was funded by Russian Federation represented by the Ministry of Science and Higher Education (Grant No. 075-15-2020-788).

\appendix

\section{Summary of modifications} \label{sec:app:differences}

Here we explicitly point out the main modification in the QKD-assisted USS scheme developed herein compared to the original scheme described in Ref.~\cite{Amiri2018}. 

\subsection{Choice of verification levels }

In the original approach, the verification level $l$ takes values from the set $\{-1,0,1,\ldots, l_{\max}\}$. 
The special $-1$th verification level is used in the majority vote dispute resolution process only.
The provided security analysis implies a secure transferability from $l=0$ to $l=-1$; however, the unforgeability is provided down to the verification level $l=0$ only. 
As described in the main text, this approach results in the emergence of the vulnerability of the majority vote process.
The idea is that there is a potential possibility for a malicious recipient to produce a message-signature pair that is accepted by honest recipients at the verification level $l=-1$ and then initiate a majority vote process with respect to this pair (recall, that unforgeability is provided down to the verification level $l=0$ only).
Then this pair will be accepted as valid in the majority vote process, though the signer may have nothing to do with it.
Therefore, the unforgeability has to be provided for all verification levels.
It can be achieved either by considering $l=-1$ in the unforgeability proof (actually, this means that $s_0$ has to be replaced by $s_{-1}$ in the expression for the probability of a forgery), or by incrementing $l_{\max}$ and considering $l\in\{0,1,\ldots,l_{\max}\}$.
Our approach is to use the latter approach in our scheme.
\subsection{Choice of $\{s_i\}$}
In the original scheme, the sequence 
\begin{equation} \label{eq:scrit_old}
    s_{l_{\max}}<s_{l_{\max}+1}<\cdots<s_{0}<s_{-1}<\frac{1}{2}
\end{equation}
is considered, while in our scheme we employ the sequence given by Eq.~\eqref{eq:scrit}.
Besides the change of the minimal verification level, there are two differences between Eqs.~\eqref{eq:scrit} and ~\eqref{eq:scrit_old}: (i) In Eq.~\eqref{eq:scrit} all $s_i$ are chosen to be equidistant and (ii) the upper bound for maximal $s_i$ is increased from $\frac{1}{2}$ to $1-2^{1-b}$.
These two modifications are made for the following reasons.

The upper bound on the probability of the non-transferability from $l=i$ to $l=i-1$ depends on $s_i-s_{i-1}$ (the higher $s_i-s_{i-1}$ is, the lower the upper bound is), so it is reasonable to have the same differences $\Delta s =s_i-s_{i-1}$ for all possible $i$, which is realized in our scheme.

The upper bound on the maximal value of $s_i$ ($s_0$ in our case) is increased from $\frac{1}{2}$ up to $1-2^{1-b}$ in order to increase the value of $\Delta s$ and so decrease the bound for a non-transferability attack.
It is achieved due to involving the value of $\epsilon$ of the employed $\epsilon$-AS2U (in our case $\epsilon=2^{1-b}$) in the upper bound on the probability of a successful forgery attack.
In contrast, in the original scheme, the worst-case scenario with $b=2$ was considered.

\subsection{Choice of $T_l$}
In the original scheme the critical number of tests for accepting the signature is given by
\begin{equation} \label{eq:Tcrit_old}
    T_l = \frac{N}{2} + (l+1)(\omega-1),
\end{equation}
while in our scheme we use values of $T_l$ from Eq.~\eqref{eq:Tcrit}.
There are three differences between Eqs.~\eqref{eq:Tcrit} and ~\eqref{eq:Tcrit_old}: The factor $(l+1)$ is replace by $l$, $N/2$ is replaced by $\omega$, and the factor $\omega-1$ is replaced by $\omega$.
The first modification comes from the consideration of $l_{\max}$ verification levels in both the security proof of non-forgeability, and transferability. 
The second modification comes from the updated security proof of non-forgeability.
We note that in our scheme $\omega<N/2$, as shown in Sec.~\ref{sec:acceptability}, so the updated bound is tighter than the original one.
The third modification comes from the consideration of non-transferability. 
In the original work a non-transferability attack is considered to be performed by a coalition which includes the signer (that is, there are no more than $\omega-1$ malicious recipients).
In our approach we consider a more general scenario, where a non-transferability attack is possible even in the case of the honest signer.
For example, the coalition of malicious recipients may try to corrupt a valid signature ${\sf Sig}_m$ for a message $m$ in such a way that ${\cal P}_i$ accepts it at the level $l$, but ${\cal P}_j$ does not accept it at any level $l'\geq l-1$. 
In our design of the scheme, the upper bound on the probability of this event drops exponentially with the parameter $k$.

\section{AS2U family construction}\label{sec:app:ASU2-construction}

Here we describe the construction of the AS2U family employed in our USS scheme.
Our construction is based on results from Ref.~\cite{Bierbrauer1994} and employs a combination of A2U and AS2U families.

\begin{definition}[A2U family].
	Let $\mathcal{A}$, $\mathcal{B}$ and $\mathcal{K}$ be finite sets. 
	A family of functions 
    $\mathcal{F} = \{f_\kappa : \mathcal{A} \rightarrow \mathcal{B}\}_{\kappa\in\mathcal{K}} $
	is called $\varepsilon$-almost 2-universal ($\varepsilon$-A2U) if
	for any distinct $m_1, m_2 \in \mathcal{A}$ and $\kappa$ picked uniformly at random from $\mathcal{K}$,
	\begin{equation}
	    \Pr[f_\kappa(m_1) = f_\kappa(m_1)] \leq \varepsilon.
	\end{equation}
\end{definition}

An AS2U family can be obtained according to the following theorem.
\begin{theorem}[composition of AS2U and A2U families \cite{Bierbrauer1994}].
    Let $\mathcal{F}_1=\{f^{(1)}_\kappa\}_{\kappa\in\mathcal{K}_1}$ be an $\varepsilon_1$-A2U family of functions from $\mathcal{A}_1$ to $\mathcal{B}_1$ and $\mathcal{F}_2=\{f^{(2)}_\kappa\}_{\kappa\in\mathcal{K}_2}$ be an $\varepsilon_2$-AS2U family of functions from $\mathcal{B}_1$ to $\mathcal{B}_2$. 
    Then the family
    \begin{equation}
        \mathcal{F}=\{f_{\kappa_1\kappa_1}\}_{\kappa_1\in\mathcal{K}_1,\kappa_2\in\mathcal{K}_2},
    \end{equation}
    with
    \begin{equation}
        f_{\kappa_1\kappa_1}(\cdot)=f^{(2)}_{\kappa_2}(f^{(1)}_{\kappa_1}(\cdot))
    \end{equation}
    is $(\varepsilon_1 + \varepsilon_2)$-AS2U.
\end{theorem}

Next we fix two integers $a$ and $b$ and construct a $2^{-b+1}$-AS2U family from $\mathcal{A}\supseteq \{0,1\}^a$ to $\mathcal{B}=\{0,1\}^b$. 
We employ one more theorem.
\begin{theorem}[relation between A2U families and error-correcting codes~\cite{Bierbrauer1994}]. \label{thm:equivalence}
    Let $\mathcal{F}=\{f_\kappa\}_{\kappa\in\mathcal{K}}$ with $\mathcal{K}=\{1,2,\ldots,K\}$ be a family of function from a finite set $\mathcal{A}$ to a finite set $\mathcal{B}$.
    The following statements are equivalent.
   \begin{enumerate}[label=(\roman*)]
        \item $\mathcal{F}$ is $\varepsilon$-A2U for some $\varepsilon>0$.
        \item The set of words $\{(f_1(m), f_2(m), \ldots, f_K(m))\}_{m\in\mathcal{A}}$ forms a code with minimum distance $\varepsilon\geq 1-d/n$.
    \end{enumerate}
\end{theorem}
One can see that an $\varepsilon$-A2U family in some sense is equivalent to an error-correcting code. 

Consider the Reed-Solomon (RS) linear error-correcting code~\cite{Reed1960}.
Let $d_{\rm RC}$, $n_{\rm RC}$, and $k_{\rm RC}$ be its minimum distance, length, and rank, respectively.

\begin{theorem}[distance of RS code~\cite{Reed1960}].
    For the RS code $d_{\rm RC}=n_{\rm RC}-k_{\rm RC}+1$.
\end{theorem}
Consider the RS code over ${\rm GF}({2^{b+s}})$ with $n_{\rm RC}=2^{b+s}$ and $k_{\rm RC}=1+2^s$.
According to Theorem~\ref{thm:equivalence} we can use the RS code to obtain a $2^{-b}$-A2U family with $\mathcal{A}={\rm GF}(2^{b+s})^{k_{\rm RC}}\cong \{0,1\}^{(b+s)(2^s+1)}$, $\mathcal{B}={\rm GF}(2^{b+s})\cong\{0,1\}^{b+s}$, and $\mathcal{K}={\rm GF}(2^{b+s})\cong\{0,1\}^{b+s}$,
where $\cong$ denotes equivalence between two sets.
In order to have $\{0,1\}^a\subseteq\mathcal{A}$ we choose a minimal possible integer $s$ such that $a\leq(b+s)(2^s+1)$.

Finally, we introduce the construction of the S2U family. 
\begin{theorem}[practical construction of S2U family~\cite{Bierbrauer1994}].
    Let $\pi: {\rm GF}(2^n)\rightarrow {\rm GF}(2^m)$ be a linear surjection. 
    Then $\mathcal{F} = \{f_{\kappa_1\kappa_2}\}_{\kappa_1\in{\rm GF}(2^n), \kappa_2\in{\rm GF}(2^m)}$ with
    $f_{\kappa_1\kappa_2}(x) = \pi(\kappa_1x)  + \kappa_2$ is S2U.
\end{theorem}

Combining the $2^{-b}$-A2U family of functions from $\mathcal{A}$ to $\{0,1\}^{b+s}$ based on RS codes and the introduced S2U family of functions from $\{0,1\}^{b+s}$ to $\{0,1\}^b$ we obtain a $2^{-b+1}$-AS2U family of functions from $\{0,1\}^a\subseteq \mathcal{A}$ to $\{0,1\}^b$.
The key length required to specify a function from the resulting family is equal to $(b+s)+(b+s+b)=3b+2s$.

\section{Proofs of the Theorems of Section~\ref{sec:security}}\label{sec:app:proofs}

\subsection{Proof of Theorem~\ref{thm:trade-off}}
\begin{proof}
    Let $\mathcal{C} \subset \{1,\ldots,N\}$ denote a subset of malicious internal recipients labels (that is any $\mathcal{P}_c$ for $c\in\mathcal{C}$ is malicious).
    Note that by definition of $\omega$, $|\mathcal{C}|\leq \omega$.
    Consider the verification process of a message-signature pair $(m,{\sf Sig}_m)$ by an honest recipient $\mathcal{P}_i$ ($i\notin \mathcal{C}$).
    According to Eq.~\eqref{eq:Tcrit} the pair is accepted at verification level $l_{\max}$ if and only if 
    \begin{equation}~\label{eq:sum-of-impacts}
        \sum_{j\notin\mathcal{C}}T^{m}_{i,j,l_{\max}}+
        \sum_{j\in\mathcal{C}}T^{m}_{i,j,l_{\max}} > \omega+l_{\max}\omega.
    \end{equation}
    Remember that, according to Eq.~\eqref{eq:scrit}, all terms $T^{m}_{i,j,l_{\max}}$ are computed by counting the number of incorrect tags and comparing the result with threshold value $s_{l_{\max}}k$.
    Each term $T^{m}_{i,j,l_{\max}}$ with $j\notin\mathcal{C}$ equals 1, since all corresponding authentication keys are obtained from the honest signer via other honest recipients.
    So, 
    \begin{equation}
        \sum_{j\notin\mathcal{C}}T^{m}_{i,j,l_{\max}}=N-\omega.
    \end{equation}
    
    The value of $\sum_{j\in\mathcal{C}}T^{m}_{i,j,l_{\max}}$ is controlled by the malicious coalition $\mathcal{C}$ and belongs to $\{0,\ldots,\omega\}$. 
    This is because each dishonest recipient can send incorrect (rubbish) keys to $\mathcal{P}_i$ at the distribution stage. 
    Considering $\sum_{j\in\mathcal{C}}T^{m}_{i,j,l_{\max}}=0$ as the worst case scenario, Eq.~\eqref{eq:sum-of-impacts} transforms into
    \begin{equation}
        N-\omega > \omega+l_{\max}\omega
    \end{equation}
    which results in Eq.~\eqref{eq:trade-off}.
    
    Any honest external recipient will also accept $(m,{\sf Sig}_m)$ at the maximal transferability level $l_{\max}$ since there will be at least $\omega+1$ honest internal recipients within the subset ${\Omega}$ used in the delegated verification.
    \end{proof}

\subsection{Proof of Theorem~\ref{thm:forgery}}

Before proceeding to the proof of Theorem~\ref{thm:forgery}, we consider the following lemma.
\begin{lemma} \label{Lemma:1}
    Let $\mathcal{F}=\{f_\kappa:\mathcal{A}\rightarrow\mathcal{B}\}_{\kappa\in\mathcal{K}}$ be an $\varepsilon$-AS2U family.
    Consider $m, m_1^\star,\ldots,m_n^\star\in \mathcal{A}$ and $t, t_1^\star, \ldots, t_n^\star\in \mathcal{B}$ such that each $m^\star_i\neq m$ and $n$ is some positive integer.
    Then
    \begin{multline}\label{eq:ASU2chain}
        \Pr[f_\kappa(m_n^\star)=t_n^\star|f_\kappa(m)=t \wedge f_\kappa(m_1^\star) \neq t_1^\star \wedge \cdots \wedge\\ f_\kappa(m_{n-1}^\star)\neq t_{n-1}^\star]
        \leq 
        (1-|\mathcal{B}|^{-1})^{n-1}\varepsilon\leq \varepsilon
    \end{multline}
    for $\kappa$ picked uniformly at random from $\mathcal{K}$.
\end{lemma}
\begin{proof}
    First of all, to simplify our consideration we introduce some new denotations.
    Let ${\bf \Psi}$, ${\bf \Psi}_i$, and $\overline{{\bf \Psi}}_i$ denote events
    \begin{equation}\label{eq:def-of-psi}
        f_\kappa(m)=t, \quad f_\kappa(m_i^\star) = t_i^\star,\quad \quad f_\kappa(m_i^\star) \neq t_i^\star,
    \end{equation}
    respectively.
    We also introduce a joint event 
    \begin{equation}\label{eq:def-of-phi}
        {\bf \Phi}_l\equiv{\bf \Psi} \wedge \overline{{\bf \Psi}}_1 \wedge \cdots \wedge \overline{{\bf \Psi}}_l.
    \end{equation}
    Then the main statement~\eqref{eq:ASU2chain} takes the compact form
    \begin{equation}
         \Pr[{\bf \Psi}_{n}|{\bf \Phi}_{n-1}]\leq (1-|\mathcal{B}|^{-1})^{n-1}\varepsilon.
    \end{equation}
    
    The proof is by induction.
    For $n=1$ Eq.~\eqref{eq:ASU2chain} directly follows from the definition of the $\varepsilon$-AS2U family.
    
    Next we assume that the main statement is true for $n=j-1$ and prove its validity for $n=j$.
    In particular, for the set $\{m_1^\star,\ldots, m_{j-2}^\star, m_{j}^\star\}$ we assume that
    \begin{equation} \label{eq:induction-assumption}
        \Pr[{\bf \Psi}_{j}|{\bf \Phi}_{j-2}]\leq (1-|\mathcal{B}|^{-1})^{j-2}\varepsilon.
    \end{equation}
    According to the Bayesian rule, we have
    \begin{multline}\label{eq:cond-probs-chain}
        \Pr[{\bf \Psi}_j|{\bf \Phi}_{j-2}]=
        \Pr[{\bf \Psi}_j|{\bf \Phi}_{j-1}]\Pr[\overline{{\bf \Psi}}_{j-1}] \pm \\
        \Pr[{\bf \Psi}_j|{\bf \Phi}_{j-2} \wedge {\bf \Psi}_{j-1}]\Pr[{\bf \Psi}_{j-1}].
    \end{multline}
    Taking into account the assumption~\eqref{eq:induction-assumption}, non-negativity of the second term on the right-hand side of~\eqref{eq:cond-probs-chain}, and the fact that $\Pr[\overline{\bf \Psi}_{j-1}]=1-\Pr[{\bf \Psi}_{j-1}]=1-|\mathcal{B}|^{-1}$ (according to the definition of the $\varepsilon$-AS2U family), we arrive at
    \begin{equation}
        \Pr[{\bf \Psi}_j|{\bf \Phi}_{j-1}] \leq (1-|\mathcal{B}|^{-1})^{j-1}\varepsilon.
    \end{equation}
    Thus Lemma 1 is proven.
\end{proof}

The results of Lemma~\ref{lemma:1} can be interpreted as follows.
Consider an authentication system for messages going from Alice to Bob and based on employing the $\varepsilon$-AS2U family $\mathcal{F}=\{f_\kappa\}_{\kappa\in\mathcal{K}}$.
Consider eavesdropper Eve, who has a valid message-authentication tag pair $(m, f_\kappa(m))$ but does not know secret key $\kappa$, and trying to forge a message from Alice to Bob.
We assume that Bob stops considering messages with tags related to the secret key $\kappa$, after either obtaining a message with valid tag or obtaining $\mu$ messages with incorrect tags.
In this way, Eve has several attempts to force Bob to accept forged message-signature pairs. 
Lemma~\ref{lemma:1} states that each next attempt of Eve has the same upper bound on the success probability ($\varepsilon$), regardless of the history of previous unsuccessful attempts.

We now proceed with the proof of Theorem~\ref{thm:forgery}

\begin{proof}
	First of all, we note that if $(m^\star,\sigma^\star)$ is rejected (at verification level $l \geq 0$) by every honest internal recipient, then $(m^\star,\sigma^\star)$ is also rejected by any honest external recipient due to construction of the delegated verification procedure (recall that the majority of any requested $2\omega+1$ internal recipients are honest).

    Let $\mathcal{C} \subset \{1,\ldots,N\}$ denote a subset of malicious internal recipient labels (that is any $\mathcal{P}_c$ for $c\in\mathcal{C}$ is malicious).
    The size of the coalition is upper bounded by $|\mathcal{C}|\leq \omega$.

    Let the coalition possess a valid message-signature pair $(m,{\sf Sig}_m)$ generated by the signer.
    In order to simplify the attackers' task as much as possible, consider an attack where the coalition tries to force an honest recipient $\mathcal{P}_i$ ($i\notin\mathcal{C}$) to accept a forged pair $(m^\star,\sigma^\star)$ with $m^\star\neq m$ at the lowest possible verification level $l=0$.
    
    The coalition possesses all authentication keys with indices $\{R_{c\rightarrow i}| c\in\mathcal{C}\}$, so they can make
    $T_{i,c,0}^{m^\star}=1$ for every $c\in\mathcal{C}$ and $m^\star$.
    Then the acceptance condition, given by Eq.~\eqref{eq:Tcrit}, takes the form
    \begin{equation}
        \sum_{j\notin\mathcal{C}}T_{i,j,0}^{m^\star}>0
    \end{equation}
    assuming the worst case scenario $|\mathcal{C}|=\omega$.
    The obtained inequality means that the coalition needs to forge at least one test corresponding to an honest recipient $j\notin\mathcal{C}$.
    Recall that a test is accepted at verification level $l=0$ if the number of tag mismatches within $k$ verified tags is less than $s_0k<k$ [see Eq.~\eqref{eq:scrit}].
 
    Now let us count the number $\mu$ of attempts that the coalition can make in order to forge a single test related to an honest recipient ${\cal P}_j$.
    First of all, each internal recipient from the coalition can try to send the forged pair directly to ${\cal P}_i$.
    If an attempt fails, then the corresponding node falls into ${\sf block\_list}_i^{\sf int}$.
    The coalition also can use  information from the delegated verification.
    In the worst case scenario all $M$ external recipients are in the coalition, so an additional $M(M+\omega)$ attempts can be made: Each external recipient ${\cal E}_i'$ can make $M+\omega$ trials until falling into ${\sf block\_list}_i^{\sf int}$ due to reaching the critical value of the counter ${\sf cnt}_{i,i'}$.
    So the resulting number of trials $\mu$ is given by
    \begin{equation}
        \mu = |{\cal C}| + M(M+\omega)\leq \omega + M(M+\omega).
    \end{equation}
    
    The success of each attempt is determined by the success of authentication tag forging. 
    According to Lemma~\ref{lemma:1} the success probability for a forgery of each tag in each attempt can be upper bounded by $\varepsilon$ of the employed $\varepsilon$-AS2U family (which is $2^{1-b}$ in our case) regardless the history of previous attempts.
    Next we use this upper bound and consider attempt results as independent variables.
    
    Let $p_0>0$ be an upper bound on a probability of event $T_{i,j,0}^{m^\star}=1$ for some $j\notin\mathcal{C}$ within a single attempt.
    Then the probability of success within $\mu$ attempts is upper bounded by
    \begin{equation}
        1-(1-p_0)^\mu <\mu p_0
    \end{equation}
    (this inequality can be verified by considering derivatives with respect to $p_0$).
    % Since there are at most $\mu$ attempts to cheat $\mathcal{P}_i$, the upper bound on a probability of making $T_{i,j,0}^{m^\star}=1$ for a fixed $j\notin\mathcal{C}$ is given by $\mu p_0$.
    The number of honest recipients equals $N-|\mathcal{C}|$, so the upper bound on the probability of a successful attack against $\mathcal{P}_i$ is given by $(N-|\mathcal{C}|)\mu p_0$.
    Moreover, there are $N-|\mathcal{C}|$ variants of choosing $i$, i.e., honest users that can be attacked, so we obtain the probability bound on the forgery event
    \begin{multline}\label{eq:forgery_with_p0}
        \Pr[{\rm forgery}]  < (N-|\mathcal{C}|)^2\mu p_0 \\< N^2[\omega+M(M+\omega)] p_0 
        = J({N,M,\omega})p_0,
    \end{multline}
    where $J({N,M,\omega})=N^2[\omega+M(M+\omega)]$.
    
    In the rest of the proof we derive an upper bound $p_0$ for probability of $T_{i,j,0}^{m^\star}=1$ for fixed $i$ and $j$.
    According to the definition of the $\varepsilon$-AS2U family we have
    \begin{equation} \label{eq:single_test_forg_bound}
        \Pr[T_{i,j,0}^{m^\star}=1] \leq \sum_{v = 0}^{\lfloor ks_0 \rfloor} \binom{k}{v} (1 - \varepsilon)^v \varepsilon^{k - v}
    \end{equation}
    with $\varepsilon=2^{-b+1}$.
    
    Consider the case $s_0<\frac{1}{2}$.
    We apply the sequence of inequalities
    \begin{multline} \label{eq:binomial_ineqs}
        (1-\varepsilon)^v\varepsilon^{k-v} = \frac{(2^{b-1}-1)^v}{2^{k(b-1)}}< \frac{2^{v(b-1)}}{2^{k(b-1)}} \\\leq
        2^{-k(b-1)(1-s_0)},
    \end{multline}
    where we use the fact that $v\leq \lfloor ks_0 \rfloor$.
    Note that the obtained result is not independent of $v$.
    Then we apply an known upper bound for a sum of binomial coefficients based on Shannon's entropy (see, e.g.,~\cite{Galvin2014} for details),
    \begin{equation}~\label{eq:Shannon_bound}
        \sum_{v=0}^{\lfloor ks_0 \rfloor} \binom{k}{v} \leq 2^{kH_2(s_0)},
    \end{equation}
    where
    \begin{equation}
        H_2(\xi) = - \xi\log_2\xi-(1-\xi)\log_2(1-\xi)
    \end{equation}
    is a binary Shannon entropy.
    Substituting the obtained bounds~\eqref{eq:binomial_ineqs} and~\eqref{eq:Shannon_bound} into Eq.~\eqref{eq:single_test_forg_bound} and then putting the result into the forgery probability bound~\eqref{eq:forgery_with_p0}, we obtain the first bound~\eqref{eq:forg-bound-1} of Theorem 2.
    
    Next consider the case $s_0<1-2^{-b+1}$.
    We can apply the well-known Hoeffding inequality to obtain
    \begin{equation}
        \sum_{v = 0}^{\lfloor ks_0 \rfloor} \binom{k}{v} (1 - \varepsilon)^v \varepsilon^{k - v} \leq e^{-2k(1 - s_0 - 2^{-b+1})^2}.
    \end{equation}
    Putting it into Eq.~\eqref{eq:single_test_forg_bound} and then into~\eqref{eq:forgery_with_p0}, we obtain the second bound~\eqref{eq:forg-bound-2} of Theorem 2.
\end{proof}

\subsection{Proof of Theorem~\ref{thm:non-transferability}}

\begin{proof}
    One can see that the non-transferability event can happen only in the situation where the malicious coalition succeeded in the creation of a message-signature pair $(m,\sigma)$ acceptable by some honest internal recipient $\mathcal{P}_i$ at verification level $l\geq1$, but not acceptable by another honest internal recipient $\mathcal{P}_j$ at verification level $l-1$.
    
    Let $\mathcal{C} \subset \{1,\ldots,N\}$ denote a subset of malicious internal recipient labels ($i,j\notin {\cal C}$).
    The coalition is able to control values of $T_{i, c, l}^m$ and $T_{j, c, l'}^m$ for any $c\in\mathcal{C}$, $m$, and $l$.
    So the best strategy for the coalition to carry out the attack is to make $T_{i, c, l}^m=1$ and $T_{j, c, l-1}^m=0$.
    Then the success of the attack corresponds to fulfillment of the following inequalities:
    \begin{equation} \label{eq:nontr-cond}
        \begin{aligned}
             &\sum_{h\notin \mathcal{C}}T_{i,h,l}^m + |\mathcal{C}|>\omega+l\omega,\\
             &\sum_{h\notin \mathcal{C}}T_{j,h,l-1}^m\leq \omega+(l-1)\omega.
        \end{aligned}
    \end{equation}
    Subtracting one from the other, we obtain
    \begin{equation} \label{eq:ineq_on_T}
        \sum_{h\notin \mathcal{C}}(T_{i,h,l}^m-T_{j,h,l-1}^m)>\omega-|\mathcal{C}|\geq 0.
    \end{equation}
    Thus, the necessary condition for Eq.~\eqref{eq:ineq_on_T} to be true is the existence of at least one $h\notin\mathcal{C}$ such that 
    \begin{equation} \label{eq:nontr-event}
        T_{i,h,l}^m=1 \wedge T_{j,h,l-1}^m=0.    
    \end{equation}
    In what follows we find an upper bound on the probability of this event. 
    Let us introduce the number tag mismatches on the side of the honest recipient $\mathcal{P}_{\tilde{i}}$ within a group of tags corresponding to some honest recipient $\mathcal{P}_h$ with respect to the message $m$,
    \begin{equation}
        G_{\tilde{i},h}^m := \sum_{r\in R_{h\rightarrow \tilde{i}}} g(f_r(m),t_r),
    \end{equation}
    where $t_r$ comes from the corresponding signature $\sigma$.
    Then the event~\eqref{eq:nontr-event} is equivalent to
    \begin{equation}
        G_{i,h}^m < s_lk \wedge 
        G_{j,h}^m \geq s_{l-1}k.
    \end{equation}
    The probability of this event can be upper bounded as follows:
    \begin{multline} \label{eq:min-of-probs}
        \Pr[G_{i,h}^m < s_lk \wedge 
        G_{j,h}^m \geq s_{l-1}k] \\ \leq \min\left(
            \Pr[G_{i,h}^m < s_lk],
            \Pr[G_{j,h}^m \geq s_{l-1}k]
        \right).
    \end{multline}
    
    The coalition is able to obtain the whole subset of tag indices $\{R_{h\rightarrow h'} | h,h'\notin \mathcal{C}\}$; however, the particular subsets $\{R_{h\rightarrow i} | h\notin \mathcal{C}\}$ and $\{R_{h\rightarrow j} | h\notin \mathcal{C}\}$ are not known to the coalition.
    Therefore, if the malicious signer corrupts some number of tags in the correct signature ${\sf Sig}_m$ of the message $m$, then the mean number of incorrect tags within $\{R_{h\rightarrow j} | h\notin \mathcal{C}\}$ and $\{R_{h\rightarrow j}\}_{h\notin \mathcal{C}}$ will be the same with respect to the distribution of introduced corruption.
    So the expectation value of $G_{i,h}^m$ and $G_{j,h}^m$ are the same.
    Let us denote it by $\overline{G}$.
    
    Then, using the Hoeffding inequality for sampling without replacement, we obtain the following bounds:
    \begin{eqnarray}
        &\begin{split}
        \Pr[G_{i,h}^m < s_lk] &\leq \Pr[|G_{i,h}^m-\overline{G}| \geq \overline{G} - s_lk ] \\ 
            &\leq 2\exp\left(
            \frac{2(s_lk - \overline{G})^2}{k}
        \right),
        \end{split}
        \\
        &\begin{split}
        \Pr[G_{j,h}^m \geq s_{l-1}k] &\leq \Pr[|G_{j,h}^m-\overline{G}| \geq s_{l-1}k - \overline{G} ] \\ &\leq 2\exp\left(
            \frac{2(s_{l-1}k - \overline{G})^2}{k}
        \right).
        \end{split}
    \end{eqnarray}
    The equality between the bounds in achieved for $\overline{G}=(s_l+s_{l-1}/2)k$, so the minimum of two probabilities in Eq.~\eqref{eq:min-of-probs} then can be bounded as
    \begin{multline}
        \min\left(
            \Pr[G_{i,h}^m < s_lk],
            \Pr[G_{j,h}^m \geq s_{l-1}k]
            \right) \\\leq 2e^{-k\Delta s^2/2},
    \end{multline}
    where $\Delta s=(s_l-s_{l-1})/2=s_0/l_{\max}$.
    
    Since there are at most $N$ honest recipients the probability of Eq.~\eqref{eq:nontr-event} for at least one $h\notin\mathcal{C}$ is upper bounded by $2Ne^{-k\Delta s^2/2}$.
    Finally, the upper bound on the probability of the non-transferability event is obtained by taking into account that there are at most $N(N-1)$ variants of choosing a pair of honest recipients $i$ and $j$.
\end{proof}

\subsection{Proof of Theorem~\ref{thm:reputication}}
\begin{proof}
    Taking into account that the number of malicious internal recipients is lower than $N/2$, the repudiation event implies a non-transferability event for the transition from $l=1$ to $l=0$.
    Thus, the probability of repudiation can be upper bounded by the probability of non-transferability given by Eq.~\eqref{eq:non-trans-bound}.
\end{proof}

\subsection{Proof of Theorem~\ref{thm:majority_vote}}
\begin{proof}
    The proof is trivial.     
    Since the majority of $2\omega+1$ internal recipients requested by ${\cal E}_i$ are honest, the outcome of the majority vote result verification will be equal to $\bot$ if there was no majority vote, and ${\sf MV}(m,\sigma)$ otherwise.
\end{proof}

\subsection{Proof of Theorem~\ref{thm:block-lists}}
\begin{proof}
To prove the theorem, we first revise situations where one (either internal or external) recipient ${\cal R}_1$ puts another (also either internal or external) recipient ${\cal R}_2$ into block list.

The first option is that a message-signature pair accepted by ${\cal R}_2$ at verification level $l\geq 1$ is sent to ${\cal R}_1$ who does not accept it at verification level  $l'\geq l-1$.
This option is equivalent to a non-transferability event, and so the probability of this false blocking is not greater than $\Pr[\text{non-transferability}]$.

The second option is that ${\cal R}_2={\cal E}_j$ is an external recipient, and an internal recipient ${\cal R}_1={\cal P}_i$ puts  ${\cal E}_j$ on the block list due to reaching a critical value of the counter ${\sf cnt}_{i,j}=M+\omega$ during the  delegated verification of a message-signature pair $(m,\sigma)$.
Let us show that if ${\cal E}_j$ is honest and there are no non-transferability and forgery events, then this cannot happen.
In the case of no non-transferability and forgery events, each incrimination of ${\sf cnt}_{i,j}$ is accompanied by adding the sender of an $(m,\sigma)$ pair to ${\cal E}_j$ to ${\sf block\_list}^{\sf ext}_j$.
Indeed, if $l_{{\rm ver},i}(m,\sigma)< l-2$, then in the case of no non-transferability or forgery event $l_{{\rm ver},i'}(m,\sigma)\leq l-2$ for every honest ${\cal P}_j$, and the result of delegated verification $l^{\rm ext}_{{\rm ver},i}\leq l-2$, and so the ${\cal E}_j$ have to add the sender to ${\sf block\_list}^{\sf ext}_j$.
The number of malicious senders can not exceed $M-1+\omega$, since all other external recipients might be malicious, and the number of malicious nodes in the internal subnetwork is less than or equal to $\omega$.
In the worst-case scenario, ${\cal E}_j$ will add all the malicious nodes to ${\sf block\_list}^{\sf ext}_j$ and still have remaining an attempt to request ${\cal P}_i$ (recall that the critical number for the counter ${\sf cnt}_{i,j}$ is equal to $M+\omega$).
The next sender of a message-signature pair will be honest, and in the absence of a non-transferability event or forgery event the counter will not increase.

In this way, the false blocking event can not happen without forgery or non-transferability events, and so the bound~\eqref{eq:false-blocking} holds true.
\end{proof}

\section{Optimization of the scheme parameters} \label{sec:app:optimization}
Here we describe the optimization procedure for finding the optimal values of the parameters $k$, $s_0$, and $b$, providing the minimal symmetric key consumption $L_{\rm tot}$, given message length $a$, security parameter $\varepsilon_{\rm tot}$, number of internal and external recipients $N$ and $M$, respectively, maximal number of dishonest nodes in the internal subnetwork $\omega$, and maximal verification level $l_{\max}$.

We first note that the bounds on probabilities of forgery and not-transferability attacks can be written as
\begin{equation}
    \begin{aligned}
            &\Pr[{\rm forgery}] < \alpha_1 e^{-\beta_1(s_0,b)k}
            \equiv B_1(k,s_0,b), \\
            &\Pr[\text{non-transferability}] \leq \alpha_2 e^{-\beta_2(s_0)k}\equiv B_2(k,s_0),
    \end{aligned}
\end{equation}
where according to~\eqref{eq:forg-bound-1},\eqref{eq:forg-bound-2}, and \eqref{eq:non-trans-bound}
\begin{equation}
    \begin{aligned}
        &\alpha_1 = N^2[\omega+M(\omega+M)], \\
        &\alpha_2 = 2N^2(N-1), \\
        &\beta_1(s_0,b) = \begin{cases}
            \max{(\beta_1'(s_0,b), \beta_1''(s_0,b))}, &s_0 < 1/2,\\\\
            \beta_1''(s_0,b), &s_0 < 1-2^{1-b},
        \end{cases}\\
         &\beta_1'(s_0,b) = (b - 1)\left(1 - s_0 - \frac{H_2(s_0)}{b - 1}\right),\\
         &\beta_1''(s_0,b)= 2(1 - s_0 - 2^{1-b})^2,\\
         &\beta_2(s_0) = s_0^2/(2l_{\max}^2).
    \end{aligned}
\end{equation}
Recall that both bounds decrease with $k$, while increasing $s_0$ affects the bounds in opposite ways.
At the same time, the total key consumption $L_{\rm tot}$ increases with $k$.

To find the optimal solution, we consider the approximate identities
\begin{equation}
    B_1(k,s_0,b) \approx B_2(k,s_0) \approx \frac{\varepsilon_{\rm tot}}{2}
\end{equation}
that corresponds to the (approximate) identity relations of bounds and saturating the tolerable security level.
These identities provide the expression for $k$ as a function of $s_0$,
\begin{equation} \label{eq:kforopt}
    k = \left\lceil \frac{\ln\alpha_2-\ln(\varepsilon_{\rm tot}/2)}{\beta_2(s_0)} \right\rceil,
\end{equation}
and the equation for $s_0$ and $b$,
\begin{equation} \label{eq:s0foropt}
    [\beta_2(s_0)-\beta_1(s_0,b)] \left\lceil \frac{\ln\alpha_2-\ln(\varepsilon_{\rm tot}/2)}{\beta_2(s_0)} \right\rceil = \ln\frac{\alpha_2}{\alpha_1}.
\end{equation}

In order to obtain the solution, we iterate over the values of $b$, solve numerically Eq.~\eqref{eq:s0foropt} to find $s_0$, then substitute it in~\eqref{eq:kforopt} to find $k$, and obtain the corresponding key consumption $L_{\rm tot}$ from Eqs.~\eqref{eq:keyconsumptions} and \eqref{eq:keyconsumptiontot}.
Typical behavior of $L_{\rm tot}$ as a function of $b$ is shown in Fig.~\ref{fig:opt-details}.
One can see that there is a clear minimum,  which is usually around the value $b=6$.
We use this value of $b$ and corresponding values of $k$ and $s_0$ as the result of the parameter optimization routine.
We also note that from Fig.~\ref{fig:opt-details} it can be seen that the consideration of the minimal possible tag length $b=2$ provides quite non-optimal results. 

\begin{figure}[]
   \centering
    \includegraphics[width=\linewidth]{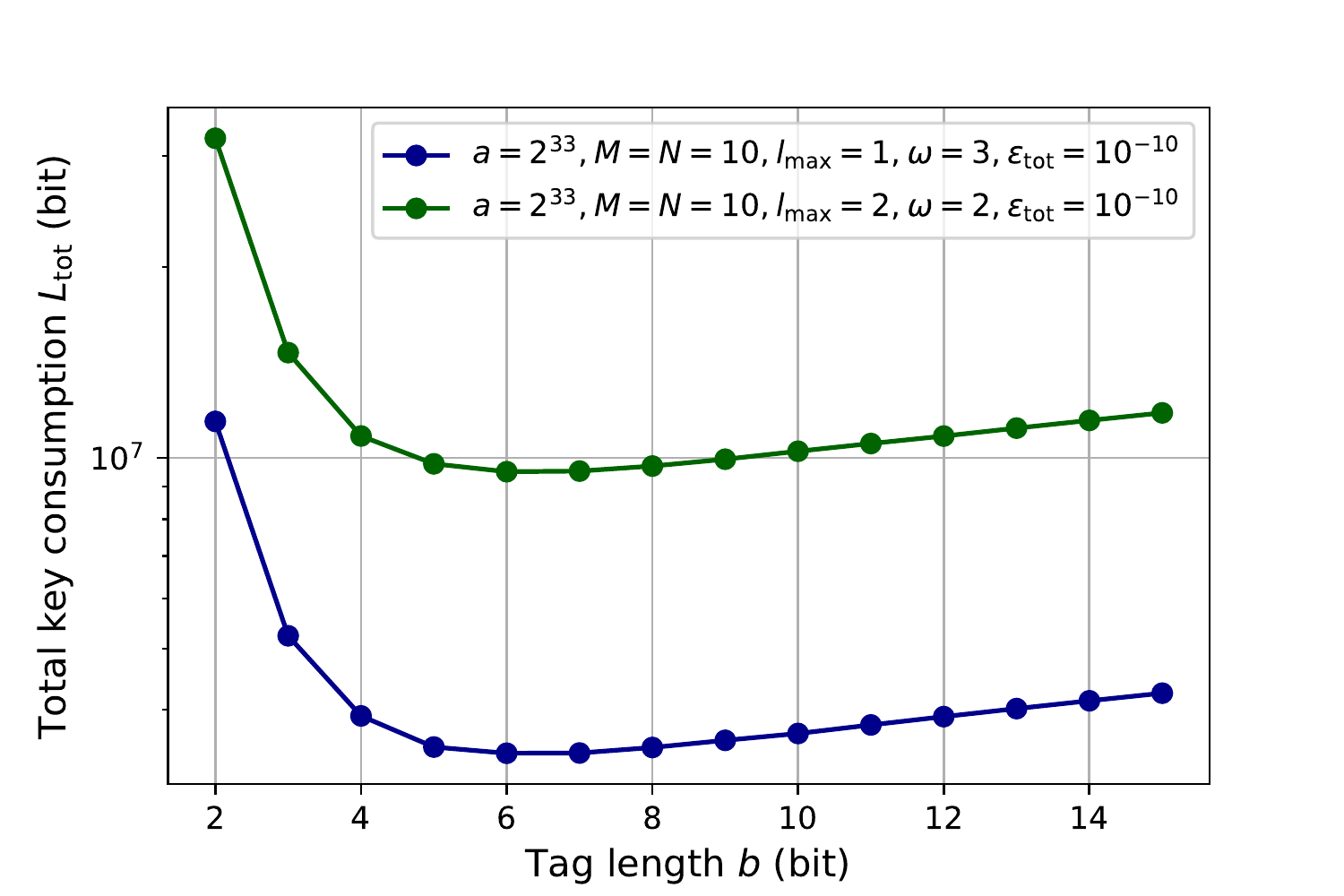}
    \vskip 2mm
    \caption{Total symmetric key consumption $L_{\rm tot}$ as a function of the tag length $b$ for two different sets of values of $a, M, N, l_{\max}, \omega$, and $\varepsilon_{\rm tot}$.
    The values of $k$ and $s_0$ are obtained from Eqs.~\eqref{eq:kforopt} and \eqref{eq:s0foropt}, respectively.}
    \label{fig:opt-details}
\end{figure}

\end{document}